%% file: wpe.tex
\documentclass[12pt]{article}
\pdfoutput=1
\usepackage{fullpage}
\usepackage{parskip}
\usepackage{amsmath,amssymb,amsthm}
\usepackage{stmaryrd}
\usepackage{wasysym} 
\usepackage{mathpartir}
\usepackage{verbatim}
\usepackage{mathpartir}
\usepackage{graphicx}
\usepackage[numbers]{natbib}
\usepackage{xcolor}
\usepackage{supertabular}

\newcommand{\BASE}{\mathsf{BASE}}
\newcommand{\DG}{\mathsf{DG}}
\newcommand{\redk}{\mathrm{red}_{\mathrm{k}}}
\newcommand{\SN}{\mathsf{SN}}
\newcommand{\SAT}{\mathsf{SAT}}
\newcommand{\T}{\mathcal{T}}
\newcommand{\interp}[2][\sigma]{\llbracket #2 \rrbracket_{#1}}
\newcommand{\interpm}[2]{\llbracket #2 \rrbracket_{#1}}
\newcommand{\subst}[2][\rho]{[ #2 ]_{#1}}
\newcommand{\substm}[2]{[ #2 ]_{#1}}

\newcommand{\ptsCtxI}[1]{\llbracket #1 \rrbracket}
\newcommand{\ptsTypI}[3][\gamma]{\llbracket #2  \vdash_{t}  #3 \rrbracket_{(#1)}}
\newcommand{\ptsTrmI}[3][\gamma]{[ #2  \vdash_{t}  #3 ]_{(#1)}}

\newcommand{\fomTrmI}[2][\Gamma]{[ #2 ]_{#1}}
\newcommand{\fomTypI}[2][\Gamma]{\llbracket #2 \rrbracket_{#1}}
\newcommand{\fomCtxI}[1]{\llbracket #1 \rrbracket}

\newcommand{\tvZ}{0}

\newtheoremstyle{better}{\parsep}{\topsep}%
     {}
     {}
     {\bfseries}
     {.}
     {0.5em}
     {} 

\theoremstyle{better}
\newtheorem{lemma}{Lemma}[section]
\newtheorem{theorem}[lemma]{Theorem}
\newtheorem{definition}[lemma]{Definition}
\newtheorem*{lemma*}{Lemma}
\newtheorem*{theorem*}{Theorem}

\newcommand{\D}{\mathcal{D}}
\newcommand{\lowered}[1]{\raisebox{-0.63\baselineskip}{\ensuremath{#1}}}
\renewcommand{\infer}[3][]{\inferrule*[right=#1]{#2}{#3}}

\newcommand{\rbracket}{]} 
\newcommand{\ext}[3]{#1[#2\mapsto #3\rbracket}

\newcommand{\sorts}{\mathcal{S}}
\newcommand{\axioms}{\mathcal{A}}
\newcommand{\rules}{\mathcal{R}}

\newcommand{\lin}{\sqsubset}

\newcommand{\EE}{\mathfrak{E}}
\newcommand{\EA}{\mathfrak{A}}
\newcommand{\EAup}[1]{\EA^{\Uparrow}(#1)}
\newcommand{\EAdn}[1]{\EA_{\Downarrow}(#1)}
\newcommand{\drop}[1]{\Downarrow_{#1}}
\newcommand{\lift}[1]{\Uparrow_{#1}}

\newcommand{\eqbig}[4]{#3 \cong^{#1}_{#2} #4}
\newcommand{\eqsmall}[4]{#3 \equiv^{#1}_{#2} #4}

\newcommand{\sarr}[2]{\Pi(#1,\,#2)}
\newcommand{\sarrdn}[2]{\Pi_{\downarrow}(#1,\,#2)}
\newcommand{\isodn}[2]{\downarrow_{\sarr{#1}{#2}}}
\newcommand{\isoup}[2]{\uparrow_{\sarr{#1}{#2}}}

\newcommand{\lev}{\mathsf{lev}}

\newcommand{\sortstar}{\ast}
\newcommand{\sortbox}{\Box}

\newif\ifcomments\commentsfalse

\ifcomments
\newcommand{\cjc}[1]{\textcolor{blue}{\textbf{[#1 ---CJC]}}}
\else
\newcommand{\cjc}[1]{}
\fi

\begin{document}
\include{wpe_inc}

\title{Strong Normalization for the Calculus of Constructions}
\author{Chris Casinghino}
\date{December 2010}
\maketitle

\section{Introduction}

The calculus of constructions (CC) is a core theory for dependently
typed programming and higher-order constructive logic.  Originally
introduced in Coquand's 1985 thesis~\cite{coquand:phd}, CC has
inspired 25 years of research in programming languages and type
theory.  Today, extensions of CC form the basis of languages like
Coq~\cite{coq-83} and Agda~\cite{norell:phd,agda}.

The popularity of CC can be attributed to the combination of its
expressiveness and its pleasant metatheoretic properties.  Among these
properties, one of the most important is {\em strong normalization},
which means that there are no infinite reduction sequences from
well-typed expressions.  This result has two important consequences.
First, it implies that CC is consistent as a logic.  This makes it an
attractive target for the formalization of mathematics.  Second, it
implies that there is an algorithm to check whether two expressions
are $\beta$-convertible.  Thus, type checking is decidable and CC
provides a practical basis for programming languages.

The strong normalization theorem has traditionally been considered
difficult to
prove~\cite{Coquand:kripkeproof,Altenkirch:realizability}.  Coquand's
original proof was found to have at least two errors, but a number of
later papers give different, correct
proofs~\cite{Coquand:kripkeproof}.  In subsequent years, many authors
considered how to extend this result for additional programming
constructs like inductive datatypes with recursion, a predicative
hierarchy of universes, and large
eliminations~\cite{werner:phd,goguen:phd,luo:book}.  Many of these
proofs are even more challenging, and several span entire theses.

This document reviews three proofs of strong normalization for CC.
Each paper we have chosen proves the theorem by constructing a model
of the system in a different domain, and each contributes something
novel to the theory of CC and its extensions.  The technical details
of the models are often complicated and intimidating.  Rather than
comprehensively verifying and reproducing the proofs, we have focused
on painting a clear picture of the beautiful and fascinating
mathematical structures that underpin them.

The first proof, originally presented by~\citet{geuversnederhof} and
subsequently popularized by~\citet{Barendregt92lambdacalculi}, models
CC in the simpler theory of F$_{\omega}$.  It demonstrates that the
strong normalization theorems for CC and F$_{\omega}$ are equivalent
by giving a reduction-preserving translation from the former to the
latter.  The second, by~\citet{geuvers:flexible}, models CC's types
with sets of expressions.  The paper demonstrates how the model may be
extended to cope with several popular language features, aiming for
flexibility.  The last proof, from~\citet{DBLP:conf/types/MelliesW96},
uses realizability semantics to consider a large class of type
theories, known as the pure type systems, which include CC.  The
authors' goal is to prove strong normalization for any pure type
system that enjoys a particular kind of realizability model.

Though each paper has a unique focus and models CC in a different
semantic system, the overall structures are very similar.  After
unifying the syntax, the correspondences between certain parts of the
proofs are quite striking.  Readers are encouraged, for example, to
compare the interpretation functions defined in
Sections~\ref{secSAT:kinds} and~\ref{secSAT:types} with those in
Section~\ref{secFOM:types}.  The similarities between the papers speak
to the fundamental underlying structure of CC, while their differences
illustrate how design choices can push the proof towards varying
goals.

The paper is structured as follows: In Sections~\ref{sec:PTS}
and~\ref{sec:PTSmetatheory} we review the definition and basic
metatheory of pure type systems and the calculus of constructions.  We
present the high-level structure of a strong normalization argument in
Section~\ref{sec:Structure}, then the proofs of
\citet{geuversnederhof}, \citet{geuvers:flexible} and
\citet{DBLP:conf/types/MelliesW96} in Sections~\ref{sec:FOM},
\ref{sec:flexible} and \ref{sec:realizability}, respectively.  We
compare the proofs and conclude in Section~\ref{sec:conclusion}.

\section{Pure Type Systems}
\label{sec:PTS}

The calculus of constructions is one example of a \textit{pure type
  system} (PTS).  This very general notion, introduced by Berardi and
popularized by Barendregt~\cite{Barendregt92lambdacalculi}, consists
of a parameterized lambda calculus which can be instantiated to a
variety of well-known type systems.  For example, the simply-typed
lambda calculus, System F, System F$_{\omega}$ and CC are all pure
type systems.  The PTS generalization is convenient because it allows
us to simultaneously study the properties of several systems.

A PTS is specified by three parameters.  First, the collection of {\it
  sorts} $s$ is given by a set $\sorts$.  The typing hierarchy among
these sorts is given by a collection of {\it axioms} $\axioms
\subseteq \sorts^2$.  Finally, the way product types may be formed is
specified by the set of {\it rules} $\rules \subseteq \sorts^3$.
Figure~\ref{fig:ptsdef} gives the complete definition of the system.

\begin{figure}
  \small
  \begin{tabular}{ll}
    $A,\,B,\,a,\,b$ & $::= 
    \ottnt{s} \;|\; 
    \ottmv{x} \;|\;
    \ottsym{(}  \ottmv{x}  \ottsym{:}  \ottnt{A}  \ottsym{)}  \to  \ottnt{B} \;|\;
    \lambda  \ottmv{x}  \ottsym{:}  \ottnt{A}  \ottsym{.}  \ottnt{b} \;|\;
    \ottnt{a} \, \ottnt{b} $\\
    $\Gamma$ & $::=  \cdot \;|\; \Gamma , \ottmv{x} \!:\! \ottnt{A} $
  \end{tabular}
  
  \bigskip

  \begin{ottdefnblock}{$\ottnt{a}  \leadsto  \ottnt{b}$}{}
  $$
  \ottdruleSBeta{}
  $$

  $$
  \ottdruleSPiOne{}\qquad \ottdruleSLamOne{}\qquad \ottdruleSAppOne{}
  $$

  $$
  \ottdruleSPiTwo{}\qquad \ottdruleSLamTwo{}\qquad \ottdruleSAppTwo{}
  $$
\end{ottdefnblock}

\begin{ottdefnblock}{$\ottnt{a}  \leadsto^{*}  \ottnt{b}$}{}
  $$
  \ottdruleMSRefl{}\qquad \ottdruleMSStep{}
  $$
\end{ottdefnblock}

\begin{ottdefnblock}{$\Gamma  \vdash  \ottnt{a}  \ottsym{:}  \ottnt{A}$}{}
  $$
  \ottdruleTSort{}\qquad \ottdruleTVar{}\qquad \ottdruleTConv{}
  $$
  
  $$
  \ottdruleTPi{}\qquad \ottdruleTLam{}\qquad \ottdruleTApp{}
  $$
\end{ottdefnblock}

\begin{ottdefnblock}{$\vdash  \Gamma$}{}
  $$
  \ottdruleCNil{}\qquad \ottdruleCCons{}
  $$
\end{ottdefnblock}

\caption{Definition of Pure Type Systems}
\label{fig:ptsdef}
\end{figure}

Choosing to explain CC as a PTS settles several questions of
presentation.  The terms, types and kinds are collapsed into
one grammar.  Some authors choose to separate these levels
syntactically for easier identification, but we find this version more
economical and it is more closely aligned with the three papers under
consideration.  For the same reasons, we have used $\beta$-conversion
in the \textsc{Conv} rule instead of using a separate judgemental
equality (as is done, for example,
in~\cite{Altenkirch:realizability}).  Here, $ =_{\beta} $ is the symmetric,
transitive, reflexive closure of $ \leadsto $.  We do not consider
$\eta$-conversion.

This context also permits a clean and compartmentalized explanation of
CC's features.  In most of the systems we consider, the sorts and
axioms are given by the sets:
$$
\sorts = \{\sortstar, \sortbox\}   \hspace{1.5em}
\axioms = \{(\sortstar,\sortbox)\}
$$
Intuitively, $\sortstar$ classifies types and $\sortbox$ classifies
kinds.  The lone axiom says that $\sortstar$ is itself a kind.  The
rule $(\sortstar,\sortstar,\sortstar)$ permits standard function
types, whose domain and codomain are both types.  The system with only
this rule is the simply-typed lambda calculus:
$$
\rules = \{  (\sortstar,\sortstar,\sortstar) \}
$$
The rule $(\sortbox,\sortstar,\sortstar)$ permits functions whose
domain is a kind.  For example, when the domain is $\sortstar$ these
are functions which takes types as arguments (i.e., polymorphism).
Thus, adding this rule yields System F:
$$
\rules = \{  (\sortstar,\sortstar,\sortstar),
           \,(\sortbox,\sortstar,\sortstar) \}
$$
The rule $(\sortbox,\sortbox,\sortbox)$ effectively duplicates STLC at
the type level.  It allows functions that take and return types.
Adding it yields System F$_{\omega}$, which has type-level
computation:
$$
\rules = \{  (\sortstar,\sortstar,\sortstar),
           \,(\sortbox,\sortstar,\sortstar),
           \,(\sortbox,\sortbox,\sortbox)\}
$$
CC adds {\it dependent types} to System F$_{\omega}$.  The rule
$(\sortstar,\sortbox,\sortbox)$ permits types to depend on terms by
allowing functions which take terms as arguments but return types.
Thus, the complete specification of CC is:
$$
\sorts = \{\sortstar, \sortbox\}   \hspace{1.5em}
\axioms = \{(\sortstar,\sortbox)\}  \hspace{1.5em}
\rules = \{  (\sortstar,\sortstar,\sortstar),
           \,(\sortbox,\sortstar,\sortstar),
           \,(\sortbox,\sortbox,\sortbox),
           \,(\sortstar,\sortbox,\sortbox)\}
$$

\section{Simple Metatheory}
\label{sec:PTSmetatheory}

For completeness, we review a few basic metatheoretic results.  We
will write $\Gamma  \vdash  \ottnt{a}  \ottsym{:}  \ottnt{A}$ for the typing judgement of an arbitrary
PTS or when it is clear what system we are discussing, and otherwise
will label the turnstile as in $\Gamma  \vdash_{CC}  \ottnt{a}  \ottsym{:}  \ottnt{A}$ for CC's typing
relation in particular.

The first result, confluence, can be proven using the standard
Tait--Martin-L\"{o}f
technique~\cite{martinlof:intuitionistic,luo:book}.

\begin{theorem}[Confluence] \label{thm:confluence}
  If $\ottnt{a}  \leadsto^{*}  \ottnt{a_{{\mathrm{1}}}}$ and $\ottnt{a}  \leadsto^{*}  \ottnt{a_{{\mathrm{2}}}}$ then there is
  a $\ottnt{b}$ such that $\ottnt{a_{{\mathrm{1}}}}  \leadsto^{*}  \ottnt{b}$ and $\ottnt{a_{{\mathrm{2}}}}  \leadsto^{*}  \ottnt{b}$.
\end{theorem}

The second property, preservation, is proved by induction on typing
derivations, using a substitution lemma.

\begin{theorem}[Preservation]
  If $\Gamma  \vdash  \ottnt{a}  \ottsym{:}  \ottnt{A}$ and $\ottnt{a}  \leadsto  \ottnt{b}$ then $\Gamma  \vdash  \ottnt{b}  \ottsym{:}  \ottnt{A}$.
\end{theorem}

The last property is not usually considered for less expressive lambda
calculi because they are presented with separate syntax for terms,
types and kinds.  The theorem says that CC expressions can still be
classified in this way with the typing judgement.  It is proved by a
straightforward induction on typing derivations.

\begin{theorem}[Classification]
  If $\Gamma  \vdash_{CC}  \ottnt{A}  \ottsym{:}  \ottnt{B}$, then exactly one of the following holds:
  \begin{itemize}
  \item $B$ is $\sortbox$.  In this case, we call $A$ a \textit{kind}.
  \item $\Gamma  \vdash_{CC}  \ottnt{B}  \ottsym{:}   \Box $.  In this case, we call $A$ a
    \textit{$\Gamma$-constructor}.
  \item $\Gamma  \vdash_{CC}  \ottnt{B}  \ottsym{:}   \ast $.  In this case, we call $A$ a 
    \textit{$\Gamma$-term}.
  \end{itemize}
\end{theorem}

When $B$ is $\sortstar$, we will call $A$ a \textit{$\Gamma$-type}.
This is a special case of the second bullet above.  In this document
we use the word ``expression'' to refer to any element of CC's grammar
and reserve the word ``term'' for the subclass of expressions
identified here.

Notice that we need a context to distinguish between constructors and
terms, but can identify kinds without one.  The ambiguity comes from
variables, and some authors avoid it by splitting them into two
syntactic classes (typically $x,y,z$ for term variables and
$\alpha,\beta$ for type variables).  Distinguishing the variables in
this way forces duplication or subtle inaccuracy when discussing
binders at different levels.  For that reason, we prefer to mix the
variables and use a context to identify the terms and constructors.

Finally, we define the central notion considered below:

\begin{definition}
  An expression is called {\it strongly normalizing} if there are no
  infinite $ \leadsto $ reduction sequences beginning at it.  We write
  $\SN$ for the collection of all such expressions.
\end{definition}

\section{Structure of the proofs}
\label{sec:Structure}

The three proofs we consider each model CC in a different domain, but
they share a similar overall structure.  In this section we describe
the technique at a high level.

\subsubsection*{Step 1: Define the interpretations}

Each proof begins by defining two interpretations.  A ``type''
interpretation, usually written $\interp[ ]{A}$, captures the static
meaning of types, kinds and sorts.  For example, in the second proof
we will model types as sets of expressions so that $\interp[ ]{A}$
contains all the terms of type $A$.  Then a ``term'' interpretation is
defined to capture the run-time behavior of terms, types and kinds.
This is usually written $[a]$.  In the example where types are
interpreted as sets of expressions, the term interpretation might pick
a canonical inhabitant with the right reduction behavior from the set.

\subsubsection*{Step 2: Relate the interpretations}

After defining the term and type interpretations, we prove a theorem
that relates them.  For example, in the second proof we will show that
if $\Gamma  \vdash  \ottnt{a}  \ottsym{:}  \ottnt{A}$, then $[a] \in \interp[ ]{A}$.  This theorem is
usually called ``soundness''.

\subsubsection*{Step 3: Declare success}

After proving the soundness theorem we observe that one of the
interpretations has some important property.  This property will mean
that strong normalization is a direct consequence of the soundness
theorem.  In the running example, $\interp[ ]{A}$ will turn out to
contain only strongly normalizing expressions.  Then, since $[a] \in
\interp[ ]{A}$ and $[a]$ models $a$'s run-time behavior, $a \in \SN$.

\subsubsection*{A clarification about the interpretations}

Though we have called $\interp[ ]{\cdot}$ the ``type'' interpretation
and $[\cdot]$ the ``term'' interpretation, we do not mean that the
former is only defined on types and the later on terms, in the sense
of the classification theorem.  Rather, $\interp[ ]{\cdot}$ is meant
to model the static meaning of any expression that can be used to
classify other expressions.  In each proof it will be defined on all
constructors, kinds and sorts of CC.  Correspondingly, $[\cdot]$ is
meant to model the dynamic behavior of any expression which can take
reduction steps.  It will be defined on the terms, constructors and
kinds of CC.

\section{Modeling CC in F$_{\omega}$}
\label{sec:FOM}

The first proof we consider translates CC expressions to System
F$_{\omega}$ in a way that preserves reduction.  System F$_{\omega}$
is known to be strongly normalizing (see~\cite{Gallier1990Girards} for
a detailed proof), so the correctness of this translation will imply
that CC is strongly normalizing as well.  The idea to prove strong
normalization of an expressive type theory by translation to a
better-understood system has been used in a variety of contexts.  For
example, \citet{Harper93framework} demonstrated that LF is strongly
normalizing by giving a reduction-preserving translation to the simply
typed lambda calculus.  This technique was originally applied to CC
by~\citet{geuversnederhof}, and their proof is reproduced
in~\citet{Barendregt92lambdacalculi}.

While this development does not have the same focus on extensibility
or generality as the later two, it has at least two advantages.
First, the proof is modular.  The other two proofs we will see are
monolithic in that they must explain the unique features of CC while
recapitulating and extending a complicated semantic argument.  Here we
may focus on the ways in which CC extends F$_{\omega}$ and can rely on
the somewhat simpler semantics of that system.  Second, the
translation itself is simple and can be verified in Peano arithmetic.
Thus, this technique demonstrates that the proof-theoretic complexity
of CC's strong normalization argument is no greater than that of
F$_{\omega}$.

\subsection{Intuition for the translation}

The calculus of constructions extends System F$_{\omega}$ with dependency in
the form of the rule $(\sortstar,\sortbox,\sortbox)$.  This rule
permits type-level abstractions which create types but take terms as
arguments.  The difficulty comes from modeling these functions in
System F$_{\omega}$ without erasing any possible reduction sequences.

To do this, we will translate expressions in two distinct ways.  The
``type'' translation $\fomTypI[]{\cdot}$ erases the dependencies to
create $F_{\omega}$ types from CC types.  The ``term'' translation
$\fomTrmI[]{\cdot}$ keeps the dependencies to avoid erasing any
possible reductions, but lowers type functions to the level of terms.
The soundness theorem of our translation will state
$$
\text{if}\;  \Gamma  \vdash_{CC}  \ottnt{a}  \ottsym{:}  \ottnt{A}  \;\text{ then }\;
    \fomCtxI{ \Gamma }   \vdash_{\text{F}_{\omega} }   \fomTrmI[ ]{ \ottnt{a} }   \ottsym{:}   \fomTypI[ ]{ \ottnt{A} } .
$$
We follow~\citet{geuversnederhof} in exhibiting how these translations
handle several examples before specifying them in full detail.
Consider a simple example of dependency where $F$ is a dependent type
function, $A$ a type and $a$ a term, so that:
$$
\infer
  {\Gamma  \vdash_{CC}  \ottnt{F}  \ottsym{:}  \ottnt{A}  \to   \ast  \\
   \Gamma  \vdash_{CC}  \ottnt{a}  \ottsym{:}  \ottnt{A}}
  {\Gamma  \vdash_{CC}  \ottnt{F} \, \ottnt{a}  \ottsym{:}   \ast }
$$
The subderivation which checks the type of $F$ will need to make use
of rule $(\sortstar,\sortbox,\sortbox)$.  We must somehow erase this
use of dependency so that, in F$_{\omega}$:
$$
\infer
  { \fomCtxI{ \Gamma }   \vdash_{\text{F}_{\omega} }   \fomTrmI[ ]{ \ottnt{F} }   \ottsym{:}   \fomTypI[ ]{ \ottnt{A}  \to   \ast  }  \\
    \fomCtxI{ \Gamma }   \vdash_{\text{F}_{\omega} }   \fomTrmI[ ]{ \ottnt{a} }   \ottsym{:}   \fomTypI[ ]{ \ottnt{A} } }
  { \fomCtxI{ \Gamma }   \vdash_{\text{F}_{\omega} }   \fomTrmI[ ]{ \ottnt{F} \, \ottnt{a} }   \ottsym{:}   \fomTypI[ ]{  \ast  } }
$$
To solve this, we take $ \fomTypI[ ]{ \ottnt{A}  \to   \ast  }  =  \fomTypI[ ]{ \ottnt{A} }   \to   \tvZ $ where
$ \tvZ  :  \ast $ is a fixed type variable that is added to the context
by $ \fomCtxI{ \Gamma } $.  We set $ \fomTypI[ ]{  \ast  }  = 0$ and $ \fomTrmI[ ]{ \ottnt{F} \, \ottnt{a} }  =  \fomTrmI[ ]{ \ottnt{F} }  \,  \fomTrmI[ ]{ \ottnt{a} } $.  Now when checking the translated $F$ we have a
term-level function rather than one which returns a type.

For our second example, suppose $A$ and $B$ are types and $a$ is a
term of type $A$.  When translating the application $\ottsym{(}  \lambda  \ottmv{x}  \ottsym{:}  \ottnt{A}  \ottsym{.}  \ottnt{B}  \ottsym{)} \, \ottnt{a}$,
we must erase $A$ to an F$_{\omega}$ type using the type translation
$ \fomTypI[ ]{ \ottnt{A} } $.  However, this admits the possibility that by erasing
dependency we will delete redexes.  This is solved by inserting an
extra redex which does nothing but provide a spot to hold $A$'s
translation as a term.  That is, for some fresh variable $y$
$$
 \fomTrmI[ ]{ \ottsym{(}  \lambda  \ottmv{x}  \ottsym{:}  \ottnt{A}  \ottsym{.}  \ottnt{B}  \ottsym{)} \, \ottnt{a} }  = \ottsym{(}  \lambda  \ottmv{y}  \ottsym{:}   \tvZ   \ottsym{.}  \lambda  \ottmv{x}  \ottsym{:}   \fomTypI[ ]{ \ottnt{A} }   \ottsym{.}   \fomTrmI[ ]{ \ottnt{B} }   \ottsym{)} \,  \fomTrmI[ ]{ \ottnt{A} }  \,  \fomTrmI[ ]{ \ottnt{a} } .
$$
The situation for polymorphism is similar.  Consider a constructor $F$
with kind $\ottsym{(}  \ottmv{x}  \ottsym{:}   \ast   \ottsym{)}  \to  \ottmv{x}  \to  \ottmv{x}$ (for example, the polymorphic
identity function) and a type $A : *$.  In translating the term $\ottnt{F} \, \ottnt{A}$, we must preserve $A$'s static meaning as a type without
erasing any possible reduction sequences.  The solution is to use
both translations, again:
\begin{center}
  \begin{tabular}{ll}
    $ \fomTypI[ ]{ \ottsym{(}  \ottmv{x}  \ottsym{:}   \ast   \ottsym{)}  \to  \ottmv{x}  \to  \ottmv{x} } $ &= $\ottsym{(}  \ottmv{x}  \ottsym{:}   \ast   \ottsym{)}  \to   \tvZ   \to  \ottmv{x}  \to  \ottmv{x}$\\
    $ \fomTrmI[ ]{ \ottnt{F} \, \ottnt{A} } $ &= $ \fomTrmI[ ]{ \ottnt{F} }  \,  \fomTypI[ ]{ \ottnt{A} }  \,  \fomTrmI[ ]{ \ottnt{A} } $
  \end{tabular}
\end{center}
The theme of these examples is that the two translations accomplish
different tasks.  The type translation $\fomTypI[ ]{\cdot}$ erases
dependencies to make F$_{\omega}$ types out of CC types.  The term
translation $\fomTrmI[ ]{\cdot}$ preserves reduction behavior but
lowers CC types to F$_{\omega}$ terms in order to accommodate the
weaker type system.  We translate parts of expressions twice so that
we can achieve both goals.

\subsection{The translation of types and contexts}
\label{secFOM:types}

Now we give the complete definition of the translation functions.  We
begin by owning up to a slight simplification in the last section.  To
distinguish term variables from type variables, the translations must
be indexed by contexts.  Thus, the translation for types becomes
$\fomTypI{\cdot}$, and the translation for terms becomes
$\fomTrmI{\cdot}$.  The translation for contexts, $\fomCtxI{\cdot}$,
remains unindexed.

In addition to these functions we define $V$ which translates CC sorts
and kinds to F$_{\omega}$ kinds:
\begin{align*}
  V( \Box )       &=  \ast \\
  V( \ast )          &=  \ast \\
  V(\ottsym{(}  \ottmv{x}  \ottsym{:}  \ottnt{A}  \ottsym{)}  \to  \ottnt{B}) &=
  \begin{cases}
    V(A) \to V(B) & \text{ if $A$ is a kind}\\
    V(B) & \text{ otherwise }
   \end{cases}
\end{align*}
This function is not indexed by a context because CC kinds may be
distinguished without one, by the classification theorem.  The reason
for the case split in the last clause is that we are erasing
dependency.

The translation of types from CC to F$_{\omega}$ follows the examples
from the previous section.  The domain of $\fomTypI{\cdot}$ is the
sorts, kinds, and $\Gamma$-constructors of CC.  We pick a unique type
variable $\tvZ$ and assume it is never used in an input to this
function.\
\begin{align*}
   \fomTypI{  \Box  }  &= \tvZ\\
   \fomTypI{  \ast  }  &= \tvZ\\
   \fomTypI{ \ottmv{x} }  &= \ottmv{x}\\
   \fomTypI{ \ottsym{(}  \ottmv{x}  \ottsym{:}  \ottnt{A}  \ottsym{)}  \to  \ottnt{B} }  &= 
  \begin{cases}
    \ottsym{(}  \ottmv{x}  \ottsym{:}   V( \ottnt{A} )   \ottsym{)}  \to   \fomTypI{ \ottnt{A} }   \to   \fomTypI[   \Gamma , \ottmv{x} \!:\! \ottnt{A}   ]{ \ottnt{B} }  & \text{ if $A$ is a kind }\\
    \ottsym{(}  \ottmv{x}  \ottsym{:}   \fomTypI{ \ottnt{A} }   \ottsym{)}  \to   \fomTypI[   \Gamma , \ottmv{x} \!:\! \ottnt{A}   ]{ \ottnt{B} }  & \text{ if $A$ is a $\Gamma$-type }
  \end{cases}\\
   \fomTypI{ \ottsym{(}  \lambda  \ottmv{x}  \ottsym{:}  \ottnt{A}  \ottsym{.}  \ottnt{B}  \ottsym{)} }  &=
  \begin{cases}
    \lambda  \ottmv{x}  \ottsym{:}   V( \ottnt{A} )   \ottsym{.}   \fomTypI[   \Gamma , \ottmv{x} \!:\! \ottnt{A}   ]{ \ottnt{B} }  & \text{ if $A$ is a kind }\\
     \fomTypI{ \ottnt{B} }  & \text{ if $A$ is a $\Gamma$-type }
  \end{cases}\\
   \fomTypI{ \ottnt{A} \, \ottnt{B} }  &=
  \begin{cases}
     \fomTypI{ \ottnt{A} }  \,  \fomTypI{ \ottnt{B} }  & \text{ if $B$ is a $\Gamma$-constructor }\\
     \fomTypI{ \ottnt{A} }        & \text{ if $B$ is a $\Gamma$-term }
  \end{cases}
\end{align*}
This function inserts duplication in product types as we discussed in
the examples section.  Otherwise, it is straightforward with the
intuition that we are erasing dependency.  The cases of the type
translation that deal with functions take into account the level of
the function's domain (just as we saw with $V$).  This distinction is
justified by the classification lemma and is reflected in the
substitution lemma for the translation:

\begin{lemma}[$\fomTypI{\cdot}$ respects substitution]
  \label{lemFOM:typsubst}
  Suppose $A$ is a kind or $\Gamma$-constructor in CC.  When $x:B \in
  \Gamma$ and $\Gamma  \vdash_{CC}  \ottnt{b}  \ottsym{:}  \ottnt{B}$, we have:
  \begin{itemize}
  \item $\fomTypI{[b/x]A} = [\fomTypI{b}/x]\fomTypI{A}$, if $B$ is a kind.
  \item $\fomTypI{[b/x]A} = \fomTypI{A}$, if $B$ is a $\Gamma$-type.
  \end{itemize}
\end{lemma}

This lemma can be shown by induction on the typing derivation.  It
follows that the translation of types preserves $\beta$-conversion:
\begin{lemma}[$\fomTypI{\cdot}$ preserves $ =_{\beta} $]
  Suppose $A$ and $A'$ are kinds or $\Gamma$-constructors in CC such
  that $\ottnt{A}  =_{\beta}  \ottnt{A'}$.  Then $ \fomTypI{ \ottnt{A} }   =_{\beta}   \fomTypI{ \ottnt{A'} } $.
\end{lemma}

Before we can state that the results of $\fomTypI{\cdot}$ are
classified by $V$, we must extend the translation to contexts.  As
mentioned, $\fomCtxI{\cdot}$ will add a type variable $\tvZ :
\sortstar$ to the context.  There are two additional changes. First, a
variable $z : \ottsym{(}  \ottmv{x}  \ottsym{:}   \ast   \ottsym{)}  \to   \ast $ will be added to help provide a
canonical inhabitant for each type.  Second, for each kind variable
$x$ which appears in $\Gamma$, the translation will add another
variable $w^x : x$.  This last change simply ensures that contexts
match up with the translation of product types, where we add an extra
argument in the case of kinds as discussed above.

We define the translation of contexts in two parts.  First, a function
$\fomTypI{x:A}$ maps each context binding to one or two translated
bindings:
$$
  \fomTypI{x:A} =
  \begin{cases}
    x:\fomTypI{A}, w^x : x & \text{ if $A$ is a $\Gamma$-kind }\\
    x:\fomTypI{A} & \text{ if $A$ is a $\Gamma$-type }
  \end{cases}
$$
The translation of a context simply maps this last function onto each
binding and adds $\tvZ$ and $z$ to the front, as mentioned.
Suppose $\Gamma = x_1 : A_1, \ldots, x_n:A_n$, then:
$$
\fomCtxI{\Gamma} = \tvZ:*,\; z:\ottsym{(}  \ottmv{x}  \ottsym{:}   \ast   \ottsym{)}  \to   \ast ,\;
                   \fomTypI{x_1:A_1},\; \ldots,\; \fomTypI{x_n:A_n}
$$
Now the soundness of the translation of types follows
straightforwardly by induction on typing derivations
\begin{lemma}[Soundness of $\fomTypI{\cdot}$]
  Suppose $A$ is a sort, kind or $\Gamma$-type of CC such that 
  $\Gamma  \vdash_{CC}  \ottnt{A}  \ottsym{:}  \ottnt{B}$.  Then $ \fomCtxI{ \Gamma }   \vdash_{\text{F}_{\omega} }   \fomTypI{ \ottnt{A} }   \ottsym{:}   V( \ottnt{B} ) $.
\end{lemma}

\subsection{The translation of terms}

As mentioned in the last section, the translation of contexts permits
the construction of a canonical inhabitant of each type or kind in F$_{\omega}$.
In particular, for any expression $B$ such that $ \fomCtxI{ \Gamma }   \vdash_{\text{F}_{\omega} }  \ottnt{B}  \ottsym{:}  \ottnt{s}$, we will define a term $c^B$ of type $B$ in the same context.
That is, $ \fomCtxI{ \Gamma }   \vdash_{\text{F}_{\omega} }   c^{ \ottnt{B} }   \ottsym{:}  \ottnt{B}$.  If $s =  \ast $, then we may use
the term $z$ to construct $ c^{ \ottnt{B} } $:
$$
 c^{ \ottnt{B} }  = \ottmv{z} \, \ottnt{B} \hfill \text{ when $B$ is a type}
$$
Otherwise, $B$ is a kind and we define:
\begin{align*}
   c^{  \ast  }  &=  \tvZ \\
   c^{ \ottsym{(}  \ottmv{x}  \ottsym{:}  \ottnt{A}  \ottsym{)}  \to  \ottnt{B} }  &= \lambda  \ottmv{x}  \ottsym{:}  \ottnt{A}  \ottsym{.}   c^{ \ottnt{B} } 
\end{align*}
The evaluation behavior of these canonical inhabitants is not very
important.  The chief purpose of $c$ is to help in the term
translation of product types.  The problem is that when $\ottsym{(}  \ottmv{x}  \ottsym{:}  \ottnt{A}  \ottsym{)}  \to  \ottnt{B}$
is a valid CC type, its translation $\ottsym{(}  \ottmv{x}  \ottsym{:}   \fomTrmI{ \ottnt{A} }   \ottsym{)}  \to   \fomTrmI{ \ottnt{B} } $ is not
necessarily well-typed in F$_{\omega}$.  The translation $\fomTypI{\cdot}$
handles this by erasing dependency, but $\fomTrmI{\cdot}$ must retain
all the possible reductions which begin at $\ottsym{(}  \ottmv{x}  \ottsym{:}  \ottnt{A}  \ottsym{)}  \to  \ottnt{B}$.  Instead
of translating it as a product, we use $c$ to construct a function
whose application to $A$ and $B$ is well-typed.  In particular, $ c^{  \tvZ   \to   \tvZ   \to   \tvZ  }  \,  \fomTrmI{ \ottnt{A} }  \,  \fomTrmI{ \ottnt{B} } $ will be a valid F$_{\omega}$ expression.  Since
$\fomTrmI{\cdot}$ does not erase the terms from $A$ and $B$, this
retains all the possible reduction sequences.

We now present the full translation of terms:

\begin{align*}
   \fomTrmI{  \ast  }  &=  c^{  \tvZ  } \\
   \fomTrmI{ \ottmv{x} }  &=
  \begin{cases}
    w^x & \text{ if $x$ is a $\Gamma$-type }\\
    x & \text{ if $x$ is a $\Gamma$-term }
  \end{cases}\\
   \fomTrmI{ \ottsym{(}  \ottmv{x}  \ottsym{:}  \ottnt{A}  \ottsym{)}  \to  \ottnt{B} }  &=
  \begin{cases}
     c^{  \tvZ   \to   \tvZ   \to   \tvZ  }  \,  \fomTrmI{ \ottnt{A} }  \, \ottsym{(}  \ottsym{[}   c^{  V( \ottnt{A} )  }   \ottsym{/}  \ottmv{x}  \ottsym{]}   [   c^{  \fomTypI{ \ottnt{A} }  }  /w^{ \ottmv{x} }]   \fomTrmI[   \Gamma , \ottmv{x} \!:\! \ottnt{A}   ]{ \ottnt{B} }    \ottsym{)}
          & \text{ if $A$ is a kind }\\
     c^{  \tvZ   \to   \tvZ   \to   \tvZ  }  \,  \fomTrmI{ \ottnt{A} }  \, \ottsym{(}  \ottsym{[}   c^{  \fomTypI{ \ottnt{A} }  }   \ottsym{/}  \ottmv{x}  \ottsym{]}   \fomTrmI[   \Gamma , \ottmv{x} \!:\! \ottnt{A}   ]{ \ottnt{B} }   \ottsym{)} 
          & \text{ if $A$ is a $\Gamma$-type }
  \end{cases}\\
   \fomTrmI{ \lambda  \ottmv{x}  \ottsym{:}  \ottnt{A}  \ottsym{.}  \ottnt{b} }  &=
  \begin{cases}
    \ottsym{(}  \lambda  \ottmv{y}  \ottsym{:}   \tvZ   \ottsym{.}  \lambda  \ottmv{x}  \ottsym{:}   V( \ottnt{A} )   \ottsym{.}   \lambda w^{ \ottmv{x} }:  \fomTypI{ \ottnt{A} }  .  \fomTrmI[   \Gamma , \ottmv{x} \!:\! \ottnt{A}   ]{ \ottnt{b} }    \ottsym{)} \,  \fomTrmI{ \ottnt{A} } 
      &\text {if $A$ is a kind, picking $y$ fresh }\\
    \ottsym{(}  \lambda  \ottmv{y}  \ottsym{:}   \tvZ   \ottsym{.}  \lambda  \ottmv{x}  \ottsym{:}   \fomTypI{ \ottnt{A} }   \ottsym{.}   \fomTrmI[   \Gamma , \ottmv{x} \!:\! \ottnt{A}   ]{ \ottnt{b} }   \ottsym{)} \,  \fomTrmI{ \ottnt{A} }  
      &\text { if $A$ is a $\Gamma$-type, picking $y$ fresh }
  \end{cases}\\
   \fomTrmI{ \ottnt{A} \, \ottnt{B} }  &=
  \begin{cases}
     \fomTrmI{ \ottnt{A} }  \,  \fomTypI{ \ottnt{B} }  \,  \fomTrmI{ \ottnt{B} }  & \text{ if $B$ is a $\Gamma$-type}\\
     \fomTrmI{ \ottnt{A} }  \,  \fomTrmI{ \ottnt{B} }  & \text{ if $B$ is a $\Gamma$-term}
  \end{cases}
\end{align*}

\begin{theorem}[Soundness of $\fomTrmI{\cdot}$]
  If $\Gamma  \vdash_{CC}  \ottnt{a}  \ottsym{:}  \ottnt{A}$ then $ \fomCtxI{ \Gamma }   \vdash_{\text{F}_{\omega} }   \fomTrmI{ \ottnt{a} }   \ottsym{:}   \fomTypI{ \ottnt{A} } $.
\end{theorem}

As we have seen with previous soundness theorems, this proof is not
conceptually surprising but requires a certain amount of book keeping.
We show only one interesting case:
\begin{proof}
  We go by induction on the structure of the derivation $\D$ of
  $\Gamma  \vdash_{CC}  \ottnt{a}  \ottsym{:}  \ottnt{A}$.
  \begin{itemize}
  \item[{\bf Case:}] $\D = \lowered{
    \infer[TLam]
     {\infer{}{\D_1 \\\\ \Gamma  \vdash_{CC}  \ottsym{(}  \ottmv{x}  \ottsym{:}  \ottnt{A}  \ottsym{)}  \to  \ottnt{B}  \ottsym{:}  \ottnt{s}} \\
      \infer{}{\D_2 \\\\  \Gamma , \ottmv{x} \!:\! \ottnt{A}   \vdash_{CC}  \ottnt{b}  \ottsym{:}  \ottnt{B}}}
     {\Gamma  \vdash_{CC}  \lambda  \ottmv{x}  \ottsym{:}  \ottnt{A}  \ottsym{.}  \ottnt{b}  \ottsym{:}  \ottsym{(}  \ottmv{x}  \ottsym{:}  \ottnt{A}  \ottsym{)}  \to  \ottnt{B}}
   }$
    
   Inversion on $\D_1$ yields a subderivation showing either that $A$
   is a kind that $A$ is a $\Gamma$-type in CC.  We will consider each
   possibility individually.  Note that because $\fomTypI{\sortbox} =
   \fomTypI{\sortstar} = 0$, in either case we have an induction
   hypothesis:
   $$
   \mathrm{IH}_A :  \fomCtxI{ \Gamma }   \vdash_{\text{F}_{\omega} }   \fomTrmI{ \ottnt{A} }   \ottsym{:}   \tvZ 
   $$
   \begin{itemize}
   \item Suppose first that $\Gamma  \vdash_{CC}  \ottnt{A}  \ottsym{:}   \ast $.  Unfolding the
     definitions of the translations, we see that we must show:
     $$
      \fomCtxI{ \Gamma }   \vdash_{\text{F}_{\omega} }  \ottsym{(}  \lambda  \ottmv{y}  \ottsym{:}   \tvZ   \ottsym{.}  \lambda  \ottmv{x}  \ottsym{:}   \fomTypI{ \ottnt{A} }   \ottsym{.}   \fomTrmI[   \Gamma , \ottmv{x} \!:\! \ottnt{A}   ]{ \ottnt{b} }   \ottsym{)} \,  \fomTrmI{ \ottnt{A} }   \ottsym{:}  \ottsym{(}  \ottmv{x}  \ottsym{:}   \fomTypI{ \ottnt{A} }   \ottsym{)}  \to   \fomTypI[   \Gamma , \ottmv{x} \!:\! \ottnt{A}   ]{ \ottnt{B} } 
     $$
     Here $y$ is some variable which doesn't occur in $\Gamma$, $A$,
     $B$ or $b$.  By IH$_A$ and the \textsc{TApp} rule, it will
     be enough to show:
     $$
      \fomCtxI{ \Gamma }   \vdash_{\text{F}_{\omega} }  \lambda  \ottmv{y}  \ottsym{:}   \tvZ   \ottsym{.}  \lambda  \ottmv{x}  \ottsym{:}   \fomTypI{ \ottnt{A} }   \ottsym{.}   \fomTrmI[   \Gamma , \ottmv{x} \!:\! \ottnt{A}   ]{ \ottnt{b} }   \ottsym{:}  \ottsym{(}  \ottmv{y}  \ottsym{:}   \tvZ   \ottsym{)}  \to  \ottsym{(}  \ottmv{x}  \ottsym{:}   \fomTypI{ \ottnt{A} }   \ottsym{)}  \to   \fomTypI[   \Gamma , \ottmv{x} \!:\! \ottnt{A}   ]{ \ottnt{B} } 
     $$
     Recall that $\tvZ : \sortstar$ will appear in $ \fomCtxI{ \Gamma } $.  By
     applying soundness for $\fomTypI{\cdot}$ to the subderivations of
     $\D_2$, we find that $ \fomTypI{ \ottnt{A} } $ and $ \fomTypI[   \Gamma , \ottmv{x} \!:\! \ottnt{A}   ]{ \ottnt{B} } $ are
     also valid F$_{\omega}$ types in the contexts $\Gamma$ and
     $ \Gamma , \ottmv{x} \!:\! \ottnt{A} $, respectively.  So by two applications of
     \textsc{TPi} and a standard weakening lemma for F$_{\omega}$, we
     have:
     $$
      \fomCtxI{ \Gamma }   \vdash_{\text{F}_{\omega} }  \ottsym{(}  \ottmv{y}  \ottsym{:}   \tvZ   \ottsym{)}  \to  \ottsym{(}  \ottmv{x}  \ottsym{:}   \fomTypI{ \ottnt{A} }   \ottsym{)}  \to   \fomTypI[   \Gamma , \ottmv{x} \!:\! \ottnt{A}   ]{ \ottnt{B} }   \ottsym{:}   \ast 
     $$
     Therefore, by rule \textsc{TLam}, it will be enough to show:
     $$
       \fomCtxI{ \Gamma }  , \ottmv{y} \!:\!  \tvZ    \vdash_{\text{F}_{\omega} }  \lambda  \ottmv{x}  \ottsym{:}   \fomTypI{ \ottnt{A} }   \ottsym{.}   \fomTrmI[   \Gamma , \ottmv{x} \!:\! \ottnt{A}   ]{ \ottnt{b} }   \ottsym{:}  \ottsym{(}  \ottmv{x}  \ottsym{:}   \fomTypI{ \ottnt{A} }   \ottsym{)}  \to   \fomTypI[   \Gamma , \ottmv{x} \!:\! \ottnt{A}   ]{ \ottnt{B} } 
     $$
     We have already observed that $\ottsym{(}  \ottmv{x}  \ottsym{:}   \fomTypI{ \ottnt{A} }   \ottsym{)}  \to   \fomTypI[   \Gamma , \ottmv{x} \!:\! \ottnt{A}   ]{ \ottnt{B} } $ is a
     valid F$_{\omega}$ type in this context.  Thus, by another application
     of \textsc{TLam}, it is sufficient to show:
     $$
        \fomCtxI{ \Gamma }  , \ottmv{y} \!:\!  \tvZ   , \ottmv{x} \!:\!  \fomTypI{ \ottnt{A} }    \vdash_{\text{F}_{\omega} }   \fomTrmI[   \Gamma , \ottmv{x} \!:\! \ottnt{A}   ]{ \ottnt{b} }   \ottsym{:}   \fomTypI[   \Gamma , \ottmv{x} \!:\! \ottnt{A}   ]{ \ottnt{B} } 
     $$
     Observe that the IH for $D_2$ is close to this (after slightly
     unfolding the interpretation of the context):
     $$
       \fomCtxI{ \Gamma }  , \ottmv{x} \!:\!  \fomTypI{ \ottnt{A} }    \vdash_{\text{F}_{\omega} }   \fomTrmI[   \Gamma , \ottmv{x} \!:\! \ottnt{A}   ]{ \ottnt{b} }   \ottsym{:}   \fomTypI[   \Gamma , \ottmv{x} \!:\! \ottnt{A}   ]{ \ottnt{B} } 
     $$
     And the result follows by a weakening lemma.
   \item Suppose instead that $\Gamma  \vdash_{CC}  \ottnt{A}  \ottsym{:}   \Box $.  After unfolding the
     translations, we must show:
     $$
      \fomCtxI{ \Gamma }   \vdash_{\text{F}_{\omega} }  \ottsym{(}  \lambda  \ottmv{y}  \ottsym{:}   \tvZ   \ottsym{.}  \lambda  \ottmv{x}  \ottsym{:}   V( \ottnt{A} )   \ottsym{.}   \lambda w^{ \ottmv{x} }:  \fomTypI{ \ottnt{A} }  .  \fomTrmI[   \Gamma , \ottmv{x} \!:\! \ottnt{A}   ]{ \ottnt{b} }    \ottsym{)} \,  \fomTrmI{ \ottnt{A} }   \ottsym{:}  \ottsym{(}  \ottmv{x}  \ottsym{:}   V( \ottnt{A} )   \ottsym{)}  \to   \fomTypI{ \ottnt{A} }   \to   \fomTypI[   \Gamma , \ottmv{x} \!:\! \ottnt{A}   ]{ \ottnt{B} } 
     $$
     Here $y$ is some fresh variable.  By IH$_A$ and rule \textsc{TApp},
     it is enough to show:
     $$
      \fomCtxI{ \Gamma }   \vdash_{\text{F}_{\omega} }  \lambda  \ottmv{y}  \ottsym{:}   \tvZ   \ottsym{.}  \lambda  \ottmv{x}  \ottsym{:}   V( \ottnt{A} )   \ottsym{.}   \lambda w^{ \ottmv{x} }:  \fomTypI{ \ottnt{A} }  .  \fomTrmI[   \Gamma , \ottmv{x} \!:\! \ottnt{A}   ]{ \ottnt{b} }    \ottsym{:}  \ottsym{(}  \ottmv{y}  \ottsym{:}   \tvZ   \ottsym{)}  \to  \ottsym{(}  \ottmv{x}  \ottsym{:}   V( \ottnt{A} )   \ottsym{)}  \to   \fomTypI{ \ottnt{A} }   \to   \fomTypI[   \Gamma , \ottmv{x} \!:\! \ottnt{A}   ]{ \ottnt{B} } 
     $$
     As before, $\tvZ$ is a $ \fomCtxI{ \Gamma } $-type and soundness for
     $\fomTypI{\cdot}$ implies that $\fomTypI{A}$ and $ \fomTypI[   \Gamma , \ottmv{x} \!:\! \ottnt{A}   ]{ \ottnt{B} } $ are valid types in F$_{\omega}$ as well.  The definition of $V$
     ensures that $V(A)$ is an F$_{\omega}$ kind.  So by several applications
     of \textsc{TPi} and weakening for F$_{\omega}$, we have:
     $$
      \fomCtxI{ \Gamma }   \vdash_{\text{F}_{\omega} }  \ottsym{(}  \ottmv{y}  \ottsym{:}   \tvZ   \ottsym{)}  \to  \ottsym{(}  \ottmv{x}  \ottsym{:}   V( \ottnt{A} )   \ottsym{)}  \to   \fomTypI{ \ottnt{A} }   \to   \fomTypI[   \Gamma , \ottmv{x} \!:\! \ottnt{A}   ]{ \ottnt{B} }   \ottsym{:}   \ast 
     $$
     Thus, by three applications of \textsc{TLam}, it is enough to
     show:
     $$
         \fomCtxI{ \Gamma }  , \ottmv{y} \!:\!  \tvZ   , \ottmv{x} \!:\!  V( \ottnt{A} )   ,w^{ \ottmv{x} }\!:\!  \fomTypI{ \ottnt{A} }    \vdash_{\text{F}_{\omega} }   \fomTrmI[   \Gamma , \ottmv{x} \!:\! \ottnt{A}   ]{ \ottnt{b} }   \ottsym{:}   \fomTypI[   \Gamma , \ottmv{x} \!:\! \ottnt{A}   ]{ \ottnt{B} } 
     $$
     This follows from the IH for $\D_2$, the observation that $ \fomCtxI{  \Gamma , \ottmv{x} \!:\! \ottnt{A}  }  =    \fomCtxI{ \Gamma }  , \ottmv{x} \!:\!  V( \ottnt{A} )   ,w^{ \ottmv{x} }\!:\!  \fomTypI{ \ottnt{A} }  $, and weakening for
     F$_{\omega}$. \qedhere
   \end{itemize}
  \end{itemize}
\end{proof}

The soundness of the term translation demonstrates that it preserves
the static semantics of CC expressions.  We must also show that it
preserves their reduction behavior.  A lemma describing the way this
function interacts with substitutions is needed.  The duplication in
the first case below mirrors the duplication we have discussed in the
translation.
\begin{lemma}[Substitution for $\fomTrmI{\cdot}$] \label{lem:fomTrmISubst}
  Suppose $\Gamma  \vdash_{CC}  \ottnt{a}  \ottsym{:}  \ottnt{A}$ and $ ( \ottmv{x} : \ottnt{B} ) \in  \Gamma $.
  \begin{itemize}
  \item If $B$ is a kind and $b$ is a $\Gamma$-type in CC, then
    $$
    \fomTrmI{[b/x]a} = [\fomTypI{b}/x][\fomTrmI{b}/w^x]\fomTrmI{a}
    $$
  \item If $B$ is a $\Gamma$-type and $b$ is a $\Gamma$-term in CC,
    then
    $$
    \fomTrmI{[b/x]a} = [\fomTrmI{b}/x]\fomTrmI{a}
    $$
  \end{itemize}
\end{lemma}

\subsection{Strong Normalization}

The final step in this proof is to relate reductions from CC
expressions with reductions from their translations.  The following
result says that the term translation does not drop any reduction
steps.

\begin{theorem}[$\fomTrmI{\cdot}$ preserves reduction]\label{lem:fomTrmRed}
  Suppose $\Gamma  \vdash_{CC}  \ottnt{a}  \ottsym{:}  \ottnt{A}$.
  $$
  \ottnt{a}  \leadsto  \ottnt{a'} \hspace{1em} \Rightarrow \hspace{1em}  \fomTrmI[  \Gamma  ]{ \ottnt{a} }   \leadsto^{*}_{\neq0}   \fomTrmI[  \Gamma  ]{ \ottnt{a'} } 
  $$
  Here, $ \leadsto^{*}_{\neq0} $ denotes reduction in at least one step.
\end{theorem}
\begin{proof}
  The proof is by induction on the derivation that $\ottnt{a}  \leadsto  \ottnt{a'}$.  The
  case of beta reduction uses Lemma~\ref{lem:fomTrmISubst}.  Each
  congruence case follows quickly by using an inversion lemma on the
  typing assumption and applying the induction hypothesis.
\end{proof}

Strong normalization for CC now follows quickly, using the same result
for F$_{\omega}$.

\begin{theorem}[Strong normalization]
  If $\Gamma  \vdash_{CC}  \ottnt{a}  \ottsym{:}  \ottnt{A}$, then $a \in \SN$.
\end{theorem}

\begin{proof}
  Assume for a contradiction that there is an infinite reduction
  sequence starting at $a$:
  $$
  \ottnt{a}  \leadsto  \ottnt{a_{{\mathrm{1}}}}  \leadsto  \ottnt{a_{{\mathrm{2}}}}  \leadsto  \ldots
  $$
  By preservation, $\Gamma  \vdash_{CC}  \ottnt{a_{\ottmv{n}}}  \ottsym{:}  \ottnt{A}$ for each $n$.  Thus, by
  Lemma~\ref{lem:fomTrmRed}, there is another infinite sequence of
  reductions:
  $$
   \fomTrmI{ \ottnt{a} }   \leadsto^{*}_{\neq0}   \fomTrmI{ \ottnt{a_{{\mathrm{1}}}} }   \leadsto^{*}_{\neq0}   \fomTrmI{ \ottnt{a_{{\mathrm{2}}}} }   \leadsto^{*}_{\neq0}  \ldots
  $$
  But by the soundness of the term interpretation, we have $ \fomCtxI{ \Gamma }   \vdash_{\text{F}_{\omega} }   \fomTrmI{ \ottnt{a} }   \ottsym{:}   \fomTypI{ \ottnt{A} } $.  This is a contradiction because the
  well-typed terms of F$_{\omega}$ are strongly normalizing.
\end{proof}

\section{Modeling types as sets of expressions}
\label{sec:flexible}

The second proof we consider, from \citet{geuvers:flexible}, will be
the most familiar to readers acquainted with the Girard--Tait method of
reducibility candidates or saturated sets.  The paper places a special
emphasis on making the proof easy to extend with additional
programming language constructs.  To this end, only the metatheory we
have introduced so far is required.\footnote{In fact, Geuvers requires
  a little less: he claims preservation isn't necessary.  He still
  relies on substitution and a strong inversion lemma, though, so our
  presentation does not deviate too far from his proof.}  Several
examples of extensions are included, and we consider some after the
development for CC itself.

\subsection{Basic definitions}

We begin with a few definitions and results relating to reduction.
Intuition for these ideas is important to understanding the main
proof, so we discuss them in some detail.

\begin{definition}\label{def:base}
  Any expression of the form $\ottmv{x} \, \ottnt{a_{{\mathrm{1}}}}\,\ldots\,\ottnt{a_{\ottmv{n}}}$ is called a {\it
    base expression}.  The set of base expressions is denoted $\BASE$.
  Note that variables are base expressions (i.e., $n = 0$ is allowed).
\end{definition}

\begin{definition}\label{def:keyred}
  With some expressions we associate another expression, called a {\it
    key redex}.
  \begin{itemize}
  \item The expression $\ottsym{(}  \lambda  \ottmv{x}  \ottsym{:}  \ottnt{A}  \ottsym{.}  \ottnt{b}  \ottsym{)} \, \ottnt{a}$ is its own key redex.
  \item If $A$ has a key redex, then $\ottnt{A} \, \ottnt{B}$ has the same key redex.
  \end{itemize}
\end{definition}

We denote by $\redk(A)$ the expression obtained by reducing $A$'s key
redex, when it has one.  Note that base expressions don't have key
redexes.  The intuition behind key redexes is that they can not be
avoided.  Reducing an expression without reducing its key redex leaves
the redex in place.  This intuition and the importance of key
reduction is captured by the following two lemmas.  They are not
difficult to prove, but they rely on a few other simple properties of
beta reduction.

\begin{lemma} \label{lem:keyredex} 
  Suppose $a$ has a key redex and $\ottnt{a}  \leadsto  \ottnt{b}$ without reducing that
  redex.  Then $b$ has a key redex, and $\redk(a)  \leadsto^{*}  \redk(b)$.
\end{lemma}

It is helpful to visualize this lemma:
$$
\begin{tabular}{ccc}
  $a$ & $ \leadsto $ & $\ottnt{b}$\\
  \rotatebox[origin=c]{270}{$ \leadsto $} & & \rotatebox[origin=c]{270}{$ \leadsto $}\\
  $\redk(a)$ & $ \leadsto^{*} $ & $\redk(b)$
\end{tabular}
$$

\begin{lemma} \label{lem:keyredexSN}
  Suppose $a,b \in \SN$ and $\redk(\ottnt{a} \, \ottnt{b}) \in \SN$.  Then $\ottnt{a} \, \ottnt{b}
  \in \SN$.
\end{lemma}

\begin{proof}
  Suppose for a contradiction that there is an infinite reduction
  sequence starting at $\ottnt{a} \, \ottnt{b}$.  Since we know $a$ and $b$ are in
  $\SN$, this means the application must beta-reduce in some finite
  number of steps.  That is, the infinite sequence has a prefix of
  the form:
  $$
  \ottnt{a} \, \ottnt{b}  \leadsto^{*}  \ottsym{(}  \lambda  \ottmv{x}  \ottsym{:}  \ottnt{A}  \ottsym{.}  \ottnt{a'}  \ottsym{)} \, \ottnt{b'}  \leadsto  \ottsym{[}  \ottnt{b'}  \ottsym{/}  \ottmv{x}  \ottsym{]}  \ottnt{a'}  \leadsto  \ldots
  $$
  Note that this last step reduces a key redex.  Thus, by multiple
  applications of lemma~\ref{lem:keyredex}, we have $\redk(\ottnt{a} \, \ottnt{b})
   \leadsto^{*}  \ottsym{[}  \ottnt{b'}  \ottsym{/}  \ottmv{x}  \ottsym{]}  \ottnt{a'}$.  This is a contradiction, since $\redk(\ottnt{a} \, \ottnt{b}) \in \SN$ but we found an infinite reduction sequence starting
  at it.
\end{proof}

Saturated sets and their closure properties are the key technical
device in the interpretation.  Originally introduced by Tait, they are
closely related to Girard's {\it candidates of reducibility} (for
detailed comparisons, see~\cite{luo:book}
and~\cite{Gallier1990Girards}).  The idea is pervasive, and we will
see it again in the second proof.

\begin{definition}
  A set of expressions $S$ is called a {\it saturated set} if the
  following three conditions hold:
  \begin{itemize}
  \item $S \subseteq \SN$
  \item $(\SN \cap \BASE) \subseteq S$
  \item If $A \in \SN$ and $\redk(A) \in S$ then $A \in S$.
  \end{itemize}
\end{definition}

The third condition states that saturated sets are ``closed under the
expansion of key redexes''.  We write $\SAT$ for the collection of all
saturated sets.  Note that $\SN \in \SAT$ and that every saturated set
is non-empty.

\begin{lemma} \label{lem:satintersection} If $S$ is a non-empty
  collection of saturated sets, then $\displaystyle\bigcap S \in \SAT$.
\end{lemma}

\begin{definition}
  If $S_1$ and $S_2$ are sets of expressions, define:
$$
\sarr{S_1}{S_2} := \{a \,|\, \forall b \in S_1, \ottnt{a} \, \ottnt{b} \in S_2 \}
$$
\end{definition}

It helps to have some intuition for this last definition: An
expression $a$ is in $\sarr{S_1}{S_2}$ if whenever $a$ is applied to
an expression in $S_1$, you get an expression in $S_2$.  Thus, when
these sets model types, $\sarr{S_1}{S_2}$ will contain the functions
from the first type to the second.  The next lemma, that
$\sarr{\cdot}{\cdot}$ preserves saturation, involves the most
intricate reasoning about reduction that appears in this proof.

\begin{lemma} \label{lem:satarrow}
  If $S_1, S_2 \in \SAT$, then $\sarr{S_1}{S_2} \in \SAT$.
\end{lemma}

\begin{proof}
  There are three conditions to verify:
  \begin{itemize}
  \item ($\sarr{S_1}{S_2} \subseteq \SN$)\; Suppose $a \in
    \sarr{S_1}{S_2}$.  Saturated sets are non-empty, so let $b \in
    S_1$ be given.  Then $\ottnt{a} \, \ottnt{b}\in S_2$, and so $\ottnt{a} \, \ottnt{b} \in \SN$.
    Thus, $\ottnt{a} \in \SN$.

  \item ($(\SN \cap \BASE) \subseteq \sarr{S_1}{S_2}$)\;  Let $a \in \SN
    \cap \BASE$ be given.  For any $b \in S_1$, since $b \in \SN$,
    $\ottnt{a} \, \ottnt{b} \in \SN \cap \BASE$.  Thus $\ottnt{a} \, \ottnt{b} \in S_2$.  So, $a \in
    \sarr{S_1}{S_2}$.

  \item ($\sarr{S_1}{S_2}$ is closed under key redex expansion)\;
    Suppose $a \in \SN$ and $\redk(a) \in \sarr{S_1}{S_2}$.  We must
    show $a \in \sarr{S_1}{S_2}$, so let $b \in S_1$ be given.  We
    have $\redk(a)\;b \in S_2$, and must show $a\;b \in S_2$.  But
    $\redk(a)\;b = \redk(\ottnt{a} \, \ottnt{b})$. Since $S_2$ is closed under
    expansion of key redexes, it is enough to show that $\ottnt{a} \, \ottnt{b} \in
    \SN$.  This follows immediately by lemma~\ref{lem:keyredexSN}.
    \qedhere
  \end{itemize}
\end{proof}

\subsection{Interpreting kinds}
\label{secSAT:kinds}

The interpretation of types comes in two steps.  First we define a
function $V(\,\cdot\,)$ on the sorts and kinds of CC.  This is,
roughly, the type of the main interpretation: if $B$ is a kind or type
such that $\Gamma  \vdash  \ottnt{B}  \ottsym{:}  \ottnt{A}$, then $B$'s interpretation will be an
element of the set $V(A)$.
\begin{align*}
  V( \Box ) = {} & \SAT\\
  V( \ast )  = {} & \SAT\\
  V(\ottsym{(}  \ottmv{x}  \ottsym{:}  \ottnt{A}  \ottsym{)}  \to  \ottnt{B}) = {} & 
      \begin{cases}
        \{f \,|\, f : V(A) \to V(B) \} & \text{when $A$ is a kind}\\
        V(B)  & \text{otherwise}
      \end{cases}
\end{align*}

By $\{f\,|\, f : V(A) \to V(B) \}$, we mean the collection of all
(set-theoretic) functions from $V(A)$ to $V(B)$.

\begin{lemma} \label{lem:Vinhabited}
  If $A$ is a kind, then $V(A)$ is non-empty.
\end{lemma}

As an example, consider the type $ \ast   \to   \ast $.  Notice that $V( \ast   \to   \ast )$ is the collection of all functions from saturated sets to
saturated sets.  So, when we interpret an expression with this type
(say $\lambda  \ottmv{x}  \ottsym{:}   \ast   \ottsym{.}  \ottmv{x}$), we will expect to get a function that takes
collections of expressions to other collections of expressions.  This
makes sense, since it is a function from types to types.

The observant reader will notice that this definition of $V$ mirrors
the one from Section~\ref{secFOM:types}.  Just as there, it indicates
that we will ignore dependency in the interpretation of types.  This
will work because of the limited ways in which CC may use terms in
types.  For example, CC lacks large eliminations: even though we can
encode natural numbers, we can not define types by pattern matching on
them.

\subsection{Interpreting types}
\label{secSAT:types}

Because our interpretation is not restricted to closed types, we begin
by defining an environment that interprets the variables.  Later, we
will consider another similar environment for terms.

\begin{definition}
  Given a context $\Gamma$ such that $\vdash  \Gamma$, a {\it constructor
    environment} $\sigma$ for $\Gamma$ is a function that maps the type
  variables of $\Gamma$ to appropriate sets according to $V$.  It
  should satisfy the condition:
  $$
  \text{if }  ( \ottmv{x} : \ottnt{A} ) \in  \Gamma \,\wedge\,\ottnt{A}\text{ is a kind, then }
  \sigma(x) \in V(A)
  $$
  We'll write $\sigma \models \Gamma$ for this relation and
  $\ext{\sigma}{x}{S}$ for the constructor environment which maps $x$ to
  $S$ and otherwise agrees with $\sigma$.
\end{definition}

Finally, we define the interpretation of types $\interp{A}$
when $A$ is a sort, kind or $\Gamma$-type:
\begin{align*}
  \interp{ \Box } = {} & \SN\\
  \interp{ \ast } = {} & \SN\\
  \interp{\ottmv{x}} = {} & \sigma(x)\\
  \interp{\ottsym{(}  \ottmv{x}  \ottsym{:}  \ottnt{A}  \ottsym{)}  \to  \ottnt{B}} = {} &
    \begin{cases}
        \sarr{\interp{A}}
             {\displaystyle\bigcap_{S \in V(A)}
                 \interp[ \ext{\sigma}{x}{S} ]{B}}  
          & \text{ when $A$ is a kind}\\
      \sarr{\interp{A}}{\interp{B}} & \text{ otherwise}
    \end{cases}\\
  \interp{\lambda  \ottmv{x}  \ottsym{:}  \ottnt{A}  \ottsym{.}  \ottnt{b}} = {} &
    \begin{cases}
        S \in V(A) \mapsto \interp[\ext{\sigma}{x}{S}]{b}
          & \text{ when $A$ is a kind}\\
        \interp{b} & \text{ otherwise}
    \end{cases}\\
  \interp{\ottnt{a} \, \ottnt{b}} = {} &
    \begin{cases}
        \interp{a}\;\interp{b}
          & \text{ when $b$ is a $\Gamma$-constructor}\\
        \interp{a} & \text{ otherwise}
    \end{cases}
\end{align*}

By $S \in V(A) \mapsto \interp[\ext{\sigma}{x}{S}]{b}$, we mean the
set-theoretic function that maps each $S$ in $V(A)$ to
$\interp[\ext{\sigma}{x}{S}]{b}$.

This type interpretation is very similar to the one from
Section~\ref{secFOM:types}.  Many of the lemmas we will need also
mirror results from that section.  For example, compare the
substitution lemma below with Lemma~\ref{lemFOM:typsubst}.

\begin{lemma}[$\interpm{ }{\cdot}$ respects substitution] 
  \label{lem:interpsubst}
  Suppose $\sigma \models \Gamma$ and $A$ is a kind or
  $\Gamma$-constructor.  When $ ( \ottmv{x} : \ottnt{B} ) \in  \Gamma $ and
  $\Gamma  \vdash  \ottnt{b}  \ottsym{:}  \ottnt{B}$, we have:
  \begin{itemize}
  \item $\interp{\ottsym{[}  \ottnt{b}  \ottsym{/}  \ottmv{x}  \ottsym{]}  \ottnt{A}} =
    \interp[\ext{\sigma}{x}{\interp{b}}]{A}$, if $B$ is a kind.
  \item $\interp{\ottsym{[}  \ottnt{b}  \ottsym{/}  \ottmv{x}  \ottsym{]}  \ottnt{A}} = \interp{A}$, if $B$
    is a $\Gamma$-type.
  \end{itemize}
\end{lemma}

From this it follows that beta-convertible types have the same
interpretation.

\begin{lemma}[$\interpm{ }{\cdot}$ respects $ =_{\beta} $]
  \label{lem:interpbeta}
  Suppose $\sigma \models \Gamma$ and $\ottnt{A_{{\mathrm{1}}}}$, $\ottnt{A_{{\mathrm{2}}}}$ are
  kinds or $\Gamma$-constructors such that $\ottnt{A_{{\mathrm{1}}}}  =_{\beta}  \ottnt{A_{{\mathrm{2}}}}$.  Then
  $\interp{\ottnt{A_{{\mathrm{1}}}}} = \interp{\ottnt{A_{{\mathrm{2}}}}}$.
\end{lemma}

As promised, the range of the interpretation is classified by the
function $V$.  The proof is by induction on typing derivations.  In
the conversion case, Lemma~\ref{lem:interpbeta} is used.

\begin{lemma}[Soundness of
  $\interpm{ }{\cdot}$] \label{lem:interptype} If $\sigma \models
  \Gamma$ and $A$ is a kind or $\Gamma$-constructor such that $\Gamma  \vdash  \ottnt{A}  \ottsym{:}  \ottnt{B}$, then $\interp{A} \in V(B)$.
\end{lemma}

An important consequence of this lemma is that the interpretation of a
type is always a saturated set and thus contains only strongly
normalizing expressions.

\subsection{From the interpretation to Strong Normalization}

The key fact about the function $\interp[ ]{\cdot}$ is that every
CC expression is in the interpretation of its type.  Before we can
prove this, we need a notion of environment for terms corresponding to
$\sigma$ for types.

\begin{definition}
  We call a mapping on variables $\rho$ a {\it term environment for
    $\Gamma$ with respect to $\sigma$} when $\sigma \models \Gamma$
  and:
  $$
  \text{ if }  ( \ottmv{x} : \ottnt{A} ) \in  \Gamma  \text{ then } \rho(x) \in \interp{A}
  $$
  We write $\sigma \models \rho : \Gamma$ for this relation and
  $\subst{A}$ for the expression created by simultaneously replacing
  the variables of $A$ with their mappings in $\rho$.  We write
  $\rho[x \mapsto a]$ for the term environment that sends $x$ to $a$
  and otherwise agrees with $\rho$.
\end{definition}

We show only the trickiest case of the key theorem---the complete
proof may be found in the appendix.  Though there are a number of
details to keep track of, all the cleverness is in the definition of
the interpretation; the result here is straightforward by induction.

\begin{theorem}[Soundness of $\substm{ }{\cdot}$]\label{thm:soundness1}
  Suppose $\Gamma  \vdash  \ottnt{a}  \ottsym{:}  \ottnt{A}$ and $\sigma \models \rho : \Gamma$.  Then
  $\subst{a} \in \interp{A}$.
\end{theorem}

\begin{proof}
  By induction on the derivation $\D$ of $\Gamma  \vdash  \ottnt{a}  \ottsym{:}  \ottnt{A}$.
  \begin{itemize}
  \item[{\bf Case:}] $\D = \lowered{
      \infer[TLam]
        {\infer{}{\D_1 \\\\ \Gamma  \vdash  \ottnt{A}  \ottsym{:}  \ottnt{s}} \\
         \infer{}{\D_2 \\\\  \Gamma , \ottmv{x} \!:\! \ottnt{A}   \vdash  \ottnt{b}  \ottsym{:}  \ottnt{B}}}
       {\Gamma  \vdash  \lambda  \ottmv{x}  \ottsym{:}  \ottnt{A}  \ottsym{.}  \ottnt{b}  \ottsym{:}  \ottsym{(}  \ottmv{x}  \ottsym{:}  \ottnt{A}  \ottsym{)}  \to  \ottnt{B}}
    }$

    The IH for $\D_1$ gives us $\subst{A} \in \SN$.  Since $x$ is
    a bound variable, we may pick it to be fresh for the domain and
    range of $\rho$.  There are two subcases: $s$ is either $\ast$ or
    $\Box$.
    \begin{itemize}
    \item Suppose $s$ is $\ast$.  Then we must show $\lambda
      x:\subst{A} . \subst{b} \in \sarr{\interp{A}}{\interp{B}}$.  So
      let $a \in \interp{A}$ be given, and observe it is enough to
      show $(\lambda x:\subst{A} . \subst{b})\;a \in \interp{B}$.

      We have $\sigma \models \rho[x \mapsto a] :  \Gamma , \ottmv{x} \!:\! \ottnt{A} $.  Thus, the
      IH for $\D_2$ gives us $\subst[\rho[x \mapsto a\rbracket ]{b} \in
      \interp{B}$.  But we also know
      $$
        (\lambda x:\subst{A} . \subst{b})\;a  \leadsto 
        [a/x]\subst{b} =
        \subst[ \rho[x \mapsto a\rbracket ]{b},
      $$
      and this step contracts a key redex.  So, it suffices to show
      that $(\lambda x:\subst{A} . \subst{b})\;a \in \SN$.  This
      follows by lemma~\ref{lem:keyredexSN}, using the classification
      lemma and lemma~\ref{lem:interptype} to show the pieces of the
      application are in $\SN$.

    \item Suppose instead that $s$ is $ \Box $.  We must show
      $\lambda x : \subst{A} . \subst{b} \in
      \sarr{\interp{A}}{\displaystyle\bigcap_{S \in V(A)}
        \interp[\ext{\sigma}{x}{S}]{B}}$.  Let an expression $a \in
      \interp{A}$ and a saturated set $S \in V(A)$ be given.  It is
      enough to show $(\lambda x:\subst{A} . \subst{b})\;a \in
      \interp[\ext{\sigma}{x}{S}]{B}$.

      Because $\ext{\sigma}{x}{S} \models \ext{\rho}{x}{a} :
       \Gamma , \ottmv{x} \!:\! \ottnt{A} $, the IH for $\D_2$ gives us that
      $\subst[\ext{\rho}{x}{a}]{b} \in
      \interp[\ext{\sigma}{x}{S}]{B}$.  As in the previous case, we
      can observe that
      $$
      (\lambda x:\subst{A} . \subst{b})\;a  \leadsto 
      [a/x] \subst{b} = \subst[\ext{\rho}{x}{a}]{b}.
      $$
      This step contracts a key redex, and by reasoning as in the last
      case we find $(\lambda x:\subst{A} . \subst{b})\;a \in
      \interp[\ext{\sigma}{x}{S}]{B}$ as desired. \qedhere
    \end{itemize}
  \end{itemize}
\end{proof}

The last result quickly implies strong normalization:

\begin{theorem}
  Suppose $\Gamma  \vdash  \ottnt{a}  \ottsym{:}  \ottnt{A}$.  Then $a \in \SN$.
\end{theorem}

\begin{proof}
  For each kind $A$, let $S_A$ be some canonical inhabitant of $V(A)$
  (by lemma~\ref{lem:Vinhabited}, these exist).  Define a constructor
  environment $\sigma$ such that, if $x : A \in \Gamma$ and $A$ is a
  kind, then $\sigma(x) \mapsto S_A$.  Define a term environment
  $\rho$ such that each variable of $\Gamma$ maps to itself.

  To see $\sigma \models \rho : \Gamma$, observe first that $\vdash  \Gamma$ (by induction on the typing derivation).  Thus each type
  assigned by $\Gamma$ itself has type $ \ast $ or $ \Box $.  So
  their interpretations are saturated sets
  (lemma~\ref{lem:interptype}), which contain all the variables.

  Thus, by the soundness of the interpretation, $a = \subst{a} \in
  \interp{A}$.  But by the classification lemma and
  lemma~\ref{lem:interptype}, $\interp{A}$ is a saturated set.  So $a
  \in \SN$.
\end{proof}

\subsection{Extensions}

We conclude the presentation of this development by describing how it
changes to accommodate several extensions to CC.  We sketch each
addition at a high level to give a sense of the proof's flexibility.
Adding small $\Sigma$-types and W-types is encouragingly
straightforward.  Unfortunately, changes at the kind level turn out to
be considerably more complicated.

\subsubsection*{Small $\Sigma$-types}

Small $\Sigma$-types classify dependent pairs where the first component is
a term.  We extend the syntax of CC with four new constructs
$$
A,B,a,b := \ldots \,|\,
 \Sigma  \ottmv{x} : \ottnt{A} . \ottnt{B}  \,|\,
\ottsym{(}  \ottnt{a}  \ottsym{,}  \ottnt{b}  \ottsym{)} \,|\,
 \mathsf{proj}_1\; \ottnt{a}  \,|\,
 \mathsf{proj}_2\;  \ottnt{a} 
$$
and we add straightforward corresponding typing rules:
$$
\infer[TSigma]
  {\Gamma  \vdash  \ottnt{A}  \ottsym{:}   \ast  \\  \Gamma , \ottmv{x} \!:\! \ottnt{A}   \vdash  \ottnt{B}  \ottsym{:}  \ottnt{s}}
  {\Gamma  \vdash   \Sigma  \ottmv{x} : \ottnt{A} . \ottnt{B}   \ottsym{:}  \ottnt{s}}
\hspace{3em}
\infer[TPair]
  {\Gamma  \vdash  \ottnt{a}  \ottsym{:}  \ottnt{A} \\ \Gamma  \vdash  \ottnt{b}  \ottsym{:}  \ottsym{[}  \ottnt{a}  \ottsym{/}  \ottmv{x}  \ottsym{]}  \ottnt{B} \\\\ \Gamma  \vdash   \Sigma  \ottmv{x} : \ottnt{A} . \ottnt{B}   \ottsym{:}  \ottnt{s}}
  {\Gamma  \vdash  \ottsym{(}  \ottnt{a}  \ottsym{,}  \ottnt{b}  \ottsym{)}  \ottsym{:}   \Sigma  \ottmv{x} : \ottnt{A} . \ottnt{B} }
$$
$$
\infer[TProj1]
  {\Gamma  \vdash  \ottnt{a}  \ottsym{:}   \Sigma  \ottmv{x} : \ottnt{A} . \ottnt{B} }
  {\Gamma  \vdash   \mathsf{proj}_1\; \ottnt{a}   \ottsym{:}  \ottnt{A}}
\hspace{3em}
\infer[TProj2]
  {\Gamma  \vdash  \ottnt{a}  \ottsym{:}   \Sigma  \ottmv{x} : \ottnt{A} . \ottnt{B} }
  {\Gamma  \vdash   \mathsf{proj}_2\;  \ottnt{a}   \ottsym{:}  \ottsym{[}   \mathsf{proj}_1\; \ottnt{a}   \ottsym{/}  \ottmv{x}  \ottsym{]}  \ottnt{B}}
$$
The reduction judgement must also change.  The obvious congruence
rules are needed for each construct, and there are two reduction rules
to handle the case where the projection operations meet pairs:
$$
\infer[SProj1] { }
  { \mathsf{proj}_1\; \ottsym{(}  \ottnt{a}  \ottsym{,}  \ottnt{b}  \ottsym{)}   \leadsto  \ottnt{a}}
\hspace{3em}
\infer[SProj2] { }
  { \mathsf{proj}_2\;  \ottsym{(}  \ottnt{a}  \ottsym{,}  \ottnt{b}  \ottsym{)}   \leadsto  \ottnt{b}}
$$
We make some simple changes to the definitions of base expressions and
key reduction.  These ensure that certain pair constructions will
always appear in our interpretations.  In particular, we extend
Definition~\ref{def:base} with the following clause:
\begin{itemize}
\item If $a \in \BASE$ then $ \mathsf{proj}_1\; \ottnt{a}  \in \BASE$ and $ \mathsf{proj}_2\;  \ottnt{a} 
  \in \BASE$.
\end{itemize}
And we extend Definition~\ref{def:keyred} with the following clause:
\begin{itemize}
\item If $a$ has a key redex, then $ \mathsf{proj}_1\; \ottnt{a} $ and $ \mathsf{proj}_2\;  \ottnt{a} $
  have the same key redex.
\end{itemize}

The definition of saturated sets remains the same, and we define a new
construction $X \otimes Y$ that is a saturated set whenever $X$ and $Y$
are:
$$
X \otimes Y := \{a\;|\;  \mathsf{proj}_1\; \ottnt{a}  \in X \wedge  \mathsf{proj}_2\;  \ottnt{a}  \in Y \}
$$
This construction is used to extend the interpretation of types for
dependent sums.  Here, since we know $x$ is a term variable, we do not
need to extend $\sigma$ in the interpretation of $B$ (just as in the
interpretation for product types and functions).
$$
\interp{ \Sigma  \ottmv{x} : \ottnt{A} . \ottnt{B} } = \interp{A} \otimes \interp{B}
$$
This new clause doesn't significantly alter the proofs of
Lemmas~\ref{lem:interpsubst}, \ref{lem:interpbeta} and
\ref{lem:interptype}.  Similarly, a quick inspection of the four new
typing rules reveals that the soundness of the term interpretation
(Theorem~\ref{thm:soundness1}) follows directly by induction in these
cases.

\subsubsection*{W-types}

W-types add well-founded trees and recursion to the calculus of
constructions.  They are common in the literature as a small change
that adds much of the expressive power of simple datatypes.  We do not
present their details, but a comprehensive introduction may be found
in~\cite{nordstrom:book}.  

Extending the proof to support W-types is only a little harder than
the previous example.  Once again, the definitions of the base
expressions and key reduction each get an extra clause.  The main
difficulty comes in defining a new construction on saturated sets to
model the $ \mathsf{W} \ottmv{x} : \ottnt{A} . \ottnt{B} $ type constructor.  The typing rule for this
constructor is:
$$
\infer[TW]
  {\Gamma  \vdash  \ottnt{A}  \ottsym{:}   \ast  \\  \Gamma , \ottmv{x} \!:\! \ottnt{A}   \vdash  \ottnt{B}  \ottsym{:}   \ast }
  {\Gamma  \vdash   \mathsf{W} \ottmv{x} : \ottnt{A} . \ottnt{B}   \ottsym{:}   \ast }
$$
The type $ \mathsf{W} \ottmv{x} : \ottnt{A} . \ottnt{B} $ classifies well-founded trees where $A$
describes the ways a tree may be formed and $B$ describes the contents
of the tree for each possible $A$.  Unsurprisingly, the interpretation
of this type involves a least fixed point over a particular monotone
operator on saturated sets.  Geuvers demonstrates that a suitable
class of operators on saturated sets has least fixed points (an
interesting exercise in set theory, but somewhat outside the scope of
the current project).

After proving this property of saturated sets, the rest of the proof
hardly changes.  An extra case is added to the interpretation of types
which uses a fixed point to interpret $ \mathsf{W} \ottmv{x} : \ottnt{A} . \ottnt{B} $.  The cases
involving the new typing rules for W-types are then straightforward by
induction.

\subsubsection*{Large $\Sigma$-types}

Large $\Sigma$-types classify dependent pairs where the first component is
a constructor.  Adding them to small $\Sigma$-types involves only one
additional typing rule:
$$
\infer[TSigmaL]
  {\Gamma  \vdash  \ottnt{A}  \ottsym{:}   \Box  \\  \Gamma , \ottmv{x} \!:\! \ottnt{A}   \vdash  \ottnt{B}  \ottsym{:}  \ottnt{s}}
  {\Gamma  \vdash   \Sigma  \ottmv{x} : \ottnt{A} . \ottnt{B}   \ottsym{:}   \Box }
$$
This addition is more complicated than small product types because the
$\otimes$ construction on saturated sets is no longer sufficient.  Defining
$$
\interp{ \Sigma  \ottmv{x} : \ottnt{A} . \ottnt{B} } = \interp{A} \otimes \interp{B}
$$
is incorrect when $A$ is a kind, because $\sigma$ must contain
interpretations for each of the type variables in $B$ on the
right-hand side.

This requires significant changes to the kind and type
interpretations.  Currently, $\interp{A}$ is a saturated set when $A$
is a kind.  Instead, $\interp{A}$ will be function from elements of
$V(A)$ to $\SAT$.  Lemma~\ref{lem:interptype} and
Theorem~\ref{thm:soundness1} change as follows:
\begin{lemma*}[Soundness of $\interpm{}{\cdot}$]
  Suppose $\sigma \models \Gamma$ and $A$ is a kind or $\Gamma$-constructor
  such that $\Gamma  \vdash  \ottnt{A}  \ottsym{:}  \ottnt{B}$.
  \begin{itemize}
  \item If $A$ is a $\Gamma$-constructor, then $\interp{A} \in V(B)$.
  \item If $A$ is a kind, then $\interp{A} \in \{f \,|\, f : V(A) \to
    \SAT\}$.    
  \end{itemize}
\end{lemma*}

\begin{theorem*}[Soundness of $\substm{}{\cdot}$]
  Suppose $\Gamma  \vdash  \ottnt{a}  \ottsym{:}  \ottnt{A}$ and $\sigma \models \rho : \Gamma$.
  \begin{itemize}
  \item If $A$ is a $\Gamma$-type, $\subst{a} \in \interp{A}$.
  \item If $A$ is a kind, $\subst{a} \in \interp{A}(\interp{a})$.
  \end{itemize}
\end{theorem*}

To illustrate these changes, we show the modified interpretations for
the kind $ \Sigma  \ottmv{x} : \ottnt{A} . \ottnt{B} $ when both $A$ and $B$ are kinds.  In this
case, we define
$$
V( \Sigma  \ottmv{x} : \ottnt{A} . \ottnt{B} ) = V(A) \times V(B)
$$
where $\times$ is the standard set-theoretic product operator.  
The type interpretation of this kind uses the new function argument
to fill in the gap we observed before:
$$
\interp{ \Sigma  \ottmv{x} : \ottnt{A} . \ottnt{B} } = 
    (X,Y) \in V(A) \times V(B) 
       \mapsto \interp{A}(X) \otimes \interp[\ext{\sigma}{x}{X}]{B}(Y)
$$
The other cases of $\interp[]{\cdot}$ that handle kinds must
also be updated, but we omit the details.  The proofs of every result
involving the kind and type interpretations must be redone, but they
are not harder.

\section{Modeling pure type systems with realizability semantics}
\label{sec:realizability}

{\it Note to the reader: The proof presented in this section is the
  most complicated of the three.  The model it uses is substantially
  more complex than the previous two, and there are several technical
  problems with the paper under consideration.  This section is
  included for completeness and to document some of the issues we
  encountered in reproducing the results.  Casual readers are
  encouraged to skip the details.}

\citet{DBLP:conf/types/MelliesW96} consider the question of strong
normalization for a subset of the pure type systems.  They define a
``realizability'' semantics parameterized in the same way as a PTS and
show that, when such a model exists, the corresponding PTS is strongly
normalizing.  Their proof identifies four specific properties that the
model must satisfy in order to guarantee strong normalization, and the
paper exhibits suitable models for several systems.  The idea of using
realizability semantics to model CC was originally introduced
by~\citet{alti:phd93}.

In this section we present their development.  The results are
particularly interesting in that the authors consider pure type
systems which are more expressive than CC.  For example, ECC (the
extended calculus of constructions) adds an infinite hierarchy of
predicative sorts to CC.  Proving strong normalization for this system
has traditionally been somewhat harder~\cite{luo:book}.

This proof is considerably more involved than the one in the previous
section.  In particular, we must define realizability models and lift
many of the ideas already explored in the context of CC to this new
domain, suitably generalized to work with any pure type system.  The
situation is additionally complicated because some of the theorems and
proofs given in the paper are false or inadequate.  We still believe
the technique is worth presenting because of its promised generality
and because it seems possible the problems here could be repaired.  We
will focus on giving intuition for the model constructions and avoid
getting caught up in the proofs.

We begin by introducing \textit{labeled} pure type systems with tight
reduction (Section~\ref{secPTS:lpts}).  The basic structures used in
the interpretation are defined in Section~\ref{secPTS:lsets}, and in
this context Section~\ref{secPTS:problem} illustrates one of the
paper's errors.  Section~\ref{secPTS:conditions} identifies the four
key properties that must hold of a model for the strong normalization
proof to apply, and examples of suitable constructions are given for
System F (Section~\ref{secPTS:systemF}) and CC
(Section~\ref{secPTS:CC}).  The interpretation of expressions into
these models is given in Section~\ref{secPTS:interp}.  Finally,
Section~\ref{secPTS:sn} discusses the paper's attempt to prove strong
normalization when a satisfactory model exists and to lift the result
back to a PTS without the extra labels.

\subsection{Labeled pure type systems and tight reduction}
\label{secPTS:lpts}

The proof we consider here is primarily concerned with
\textit{labeled} pure type systems which have more type annotations
and a restricted reduction relation.  In particular, the syntactic
forms for function abstraction ($ \lambda _{( \ottmv{x} : A ) \to B }. b $) and application
($ \mathsf{app} _{( \ottmv{x} : A ) \to B } ( b , a ) $) are now labeled with the complete type
of the function involved.  Additionally, the rule for beta reduction
has been modified to demand that the annotations match:
$$
\ottdruleTSBeta{}
$$
This new rule is known as \textit{tight reduction}.  The restrictions
give us more information about types from the syntax itself, which can
help to avoid potential circularity in the proof.

The complete specification of the modified system can be found in
Figure~\ref{fig:ptstight}.  The only other significant change is in
the conversion rule, which demands that one type is actually reducible
to the other.  This ensures that conversions take place in a path
through the set of well-typed expressions.  After proving that the
expressions of this system are strongly normalizing, it will be
relatively simple to lift the result back to a standard PTS by adding
annotations.

In what follows we consider an arbitrary labeled PTS and let
$\T$ stand for the set of its expressions.  Many of our definitions
from the previous section easily adapt to the new domain.  The base
expressions are now those with the form:
$$
  \mathsf{app}_{(y_n:A_n)\to B_n}(\ldots( \mathsf{app} _{( \ottmv{y_{{\mathrm{1}}}} : A_{{\mathrm{1}}} ) \to B_{{\mathrm{1}}} } ( \ottmv{x} , a_{{\mathrm{1}}} ) , \ldots), a_n)
$$

The definitions of key redexes, saturated sets and the product
construction $\sarr{S_1}{S_2}$ on saturated sets remain the same.
Syntactically, key reductions now look like this:
\begin{center}
  \begin{tabular}{cl}
              & 
      $\mathsf{app}_{(x_n:A_n)\to B_n} (\ldots 
             ( \mathsf{app} _{( \ottmv{x} : A ) \to B } (  \lambda _{( \ottmv{x} : A ) \to B }. b  , a ) , a_1),
                  \ldots),a_n)$ \\
    $ \leadsto_{t} $ &
      $\mathsf{app}_{(x_n:A_n)\to B_n} ( \ldots (\mathsf{app}_{(x_1:A_1)\to B_1} 
                                   (\ottsym{[}  a  \ottsym{/}  \ottmv{x}  \ottsym{]}  b, a_1), \ldots), a_n)$
  \end{tabular}
\end{center}
As before, we'll write $a = \redk(b)$ when $a$ is the labeled expression
that results from reducing $b$'s key redex.

\begin{figure}
  \small
  \begin{tabular}{ll}
    $A,\,B,\,a,\,b$ & $::= 
    \ottnt{s} \;|\; 
    \ottmv{x} \;|\;
    \ottsym{(}  \ottmv{x}  \ottsym{:}  A  \ottsym{)}  \to  B \;|\;
     \lambda _{( \ottmv{x} : A ) \to B }. b  \;|\;
     \mathsf{app} _{( \ottmv{x} : A ) \to B } ( a , b )  $\\
    $\Gamma$ & $::=  \cdot \;|\; \Gamma , \ottmv{x} \!:\! A $
  \end{tabular}
  
  \bigskip

  \begin{ottdefnblock}{$a  \leadsto_{t}  b$}{}
  $$
  \ottdruleTSBeta{}
  $$
  \medskip
  \begin{center}
    plus contextual rules, including reduction in the type annotations
  \end{center}
\end{ottdefnblock}
\bigskip
\begin{ottdefnblock}{$a  \leadsto^{*}_{t}  b$}{}
  $$
  \ottdruleMTSRefl{}\qquad \ottdruleMTSStep{}
  $$
\end{ottdefnblock}

\begin{ottdefnblock}{$\Gamma  \vdash_{t}  a  \ottsym{:}  A$}{}
  $$
  \ottdruleTTSort{}\qquad \ottdruleTTVar{}
  $$
  
  $$
  \ottdruleTTConv{} \qquad \ottdruleTTPi{}
  $$

  $$
  \ottdruleTTLam{}\qquad \ottdruleTTApp{}
  $$
\end{ottdefnblock}

\begin{ottdefnblock}{$\vdash_{t}  \Gamma$}{}
  $$
  \ottdruleTCNil{}\qquad \ottdruleTCCons{}
  $$
\end{ottdefnblock}

\caption{Definition of PTS with tight reduction}
\label{fig:ptstight}
\end{figure}

\subsection{Realizability constructions}
\label{secPTS:lsets}

We now define the basic constructions that will be used in the
interpretation of a PTS.
\begin{definition}
  A \textit{$\Lambda$-set} is a pair $(X_0,\models)$ where $X_0$ is
  any set and $\cdot \models \cdot \subseteq \T \times X_0$ is a
  relation between the set of labeled expressions and $X_0$.

  We call the elements of $X_0$ the \textit{carriers}.  When $\alpha
  \in X_0$, the \textit{realizers} of $\alpha$ are the expressions $A$
  such that $A \models \alpha$.  When $X$ is a $\Lambda$-set we write
  $X_0$ for its first component, $\models_X$ for its second, and
  $\alpha \lin X$ for $\alpha \in X_0$.
\end{definition}

Roughly speaking, a type $B$ will be modeled as a $\Lambda$-set whose
realizers include the terms of type $B$.  We can think of
$\Lambda$-sets as sets of expressions with some extra structure
provided by the carriers.  Like sets of expressions, $\Lambda$-sets
can be saturated:
\begin{definition}
  A $\Lambda$-set $X$ is \textit{saturated} if:
  \begin{itemize}
  \item every realizer is strongly normalizing,
  \item there is a carrier that is realized by every element of
    $\BASE\cap\SN$, and
  \item if $\alpha \lin X$, the realizers of $\alpha$ are closed under
    the expansion of key redexes.  That is, if $a \models \alpha$,
    $a = \redk(b)$ and $b \in \SN$, then $b \models \alpha$.
  \end{itemize}
\end{definition}

It is not hard to see that if $X$ is a saturated $\Lambda$-set, $X$'s
realizers form a saturated set.  We also identify a class of
isomorphisms between $\Lambda$-sets:
\begin{definition}
  Let $X$ and $Y$ be two $\Lambda$-sets.  A \textit{$\Lambda$-iso} $p$
  from $X$ to $Y$ is a bijective function $p : X_0 \to Y_0$ such that
  $a \models_X \alpha$ iff $a \models_Y p(\alpha)$.
\end{definition}

As we suggested earlier, types will be modeled by $\Lambda$-sets.
Unsurprisingly, then, sorts will be modeled by collections of
$\Lambda$-sets.  We introduce some additional structure in these
collections to deal with the circularity in some pure type systems.
In the definitions to follow we have some fixed set $\EE$ which will
index a family of equivalence relations.  We will instantiate $\EE$
when we give models for particular theories.
\begin{definition}
  An {\it $\EE$-set} $\EA$ is a set of $\Lambda$-sets that is paired with
  two families of equivalence relations, indexed by $i \in \EE$:
  $$
    \eqbig{i}{\EA}{\cdot}{\cdot} \subseteq \EA^2
    \hspace{3em} \text{and} \hspace{3em}
    \eqsmall{i}{\EA}{\cdot}{\cdot} \subseteq 
      (\displaystyle\bigcup_{X \in \EA} X_0)^2
  $$
\end{definition}
\vspace{-1em} Here, $\eqbig{i}{\EA}{}{}$ equates some of the
$\Lambda$-sets in $\EA$, and $\eqsmall{i}{\EA}{}{}$ equates some of
their carriers.  As we will see, it would be hard to interpret types
formed using CC's rule $(\sortbox,\sortbox,\sortbox)$ inside of set
theory without breaking up the $\EE$-set associated with $\sortbox$
into equivalence classes.

We can now lift the notion of products to $\Lambda$-sets.  When
interpreting a type $\ottsym{(}  \ottmv{x}  \ottsym{:}  A  \ottsym{)}  \to  B$, we will have a $\Lambda$-set
$X$ for $A$ and a family of $\Lambda$-sets for $B$, one for
each carrier of $X$.  The following construction defines a new
corresponding $\Lambda$-set.
\begin{definition}[$\Lambda$-set products]
  Let $\EA_1$ and $\EA_2$ be $\EE$-sets.  Suppose $X \in \EA_1$ and
  $Y_{\alpha} \in \EA_2$ for each $\alpha \lin X$.  We define a new
  $\Lambda$-set $\sarr{X}{Y}$ by:

  \begin{tabular}{rcl}
    $\sarr{X}{Y}_0$ & := &
      $\{ f : (\alpha \in X_0) \to (Y_\alpha)_0 \,|\,
           \forall \alpha,\alpha' \in X_0, \forall i \in \EE,\,
                          \eqsmall{i}{\EA_1}{\alpha}{\alpha'}
             \Rightarrow
             \eqsmall{i}{\EA_2}{f(\alpha)}{f(\alpha')}\}$\\
             \\
    $a \models_{\sarr{X}{Y}} f$ & iff &
      $\forall \alpha \in X_0, \forall b \models_{X} \alpha, 
       \forall A,B \in \SN,\, 
            \mathsf{app} _{( \ottmv{x} : A ) \to B } ( a , b )  \models_{Y_{\alpha}} f(\alpha)$
  \end{tabular}
\end{definition}

The carriers of this new $\Lambda$-set are set-theoretic functions
which take any carrier $\alpha \lin X$ to a carrier of $Y_{\alpha}$.
These functions must map carriers related in $\EA_1$ to carriers
related in $\EA_2$.  Intuitively, a term $a$ realizes such a function
$f$ when any application of $f$ is realized by the corresponding
applications of $a$.

Our last definition in this section extends the equivalence relations
of two $\EE$-sets to the products between them.  The definition is
somewhat intricate and can be skipped on a first reading.
\begin{definition}
  Let $\EA_1$ and $\EA_2$ be $\EE$-sets.  Suppose we have $X,X' \in
  \EA_1$ and two corresponding families of $\Lambda$-sets,
  $Y_{\alpha},Y'_{\alpha'} \in \EA_2$ for each $\alpha \in X$ and
  $\alpha' \in X'$.  For each $i \in \EE$, we define two new
  relations:
  \begin{itemize}
  \item$\eqbig{i}{\sarr{\EA_1}{\EA_2}}{\sarr{X}{Y}}{\sarr{X'}{Y'}}$
    \hspace{0.5em} iff:
    $$
         \eqbig{i}{\EA_1}{X}{X'}
      \hspace{2em}\text{and}\hspace{2em}
         \forall (\alpha,\alpha') \in X_0 \times X'_0,\;
                         \eqsmall{i}{\EA_1}{\alpha}{\alpha'} 
            \Rightarrow  \eqbig{i}{\EA_2}{Y_{\alpha}}{Y'_{\alpha'}}
    $$
  \item Suppose $\eqbig{i}{\sarr{\EA_1}{\EA_2}}{\sarr{X}{Y}}{\sarr{X'}{Y'}}$.
      When $f \lin \sarr{X}{Y}$ and $g \lin \sarr{X'}{Y'}$
      we define $\eqsmall{i}{\sarr{\EA_1}{\EA_2}}{f}{g}$ iff:
      $$
      \forall (\alpha,\alpha') \in X_0 \times X'_0,\;
                    \eqsmall{i}{\EA_1}{\alpha}{\alpha'}
        \Rightarrow \eqsmall{i}{\EA_2}{f(\alpha)}{g(\alpha')}
      $$
  \end{itemize}
\end{definition}

\subsection{A problem}
\label{secPTS:problem}

As mentioned in the introduction, there are several problems with this
proof.  The first comes from the definition of $\Lambda$-set products.
The authors claim to prove that this operation preserves saturation:
\begin{quote}
  Let $\EA_1$ and $\EA_2$ be $\EE$-sets.  Suppose $X \in \EA_1$ and
  $Y_{\alpha} \in \EA_2$ for each $\alpha \lin X$.  If $X$ and each
  $Y_{\alpha}$ are saturated $\Lambda$-sets, then so is $\sarr{X}{Y}$.
\end{quote}

While a similar result holds for for saturated sets, this proposition
is false.  A saturated $\Lambda$-set must have a carrier which
realizes every element of $\BASE \cap \SN$, but $\sarr{X}{Y}$ may have
no carriers at all.

As an example, consider three $\Lambda$-sets, $X, Y_{\beta}$ and
$Y_{\gamma}$, such that $X$ has two carriers and the others have one,
and each carrier is realized by every strongly normalizing expression.
That is:
\begin{eqnarray*}
  X :=& (\{\beta,\gamma\},\,\SN\times\{\beta,\gamma\})\\
  Y_{\beta} :=& (\{\beta'\},\,\SN\times\{\beta'\})\\
  Y_{\gamma} :=& (\{\gamma'\},\,\SN\times\{\gamma'\})
\end{eqnarray*}
These $\Lambda$-sets are saturated.  Suppose $\EE$ is a singleton set
$\{1\}$, and define two $\EE$-sets:
\begin{center}
  \begin{tabular}{lclclcl}
    $\EA_1$ & $:=$ & $\{X\}$ &  \hspace{4em} &$\EA_2$ & $:=$ & $ \{Y_{\beta},\,Y_{\gamma}\}$\\
    $\eqbig{1}{\EA_1}{}{}$ & $ :=$ & $ \{(X,X)\}$ & &
    $\eqbig{1}{\EA_2}{}{}$ & $ :=$ & $ \EA_2 \times \EA_2 $ \\
    $\eqsmall{1}{\EA_1}{}{}$ &$ :=$ & $ X_0 \times X_0$ & &
    $\eqsmall{1}{\EA_2}{}{}$ &$ :=$ & $ \{(\beta',\beta'),(\gamma',\gamma')\}$
  \end{tabular}
\end{center}
Notice in particular that $\beta$ and $\gamma$ are related by
$\eqsmall{1}{\EA_1}{}{}$ but that $\beta'$ and $\gamma'$ are not
related by $\eqsmall{1}{\EA_2}{}{}$.  Any carrier of $\sarr{X}{Y}$
would have to map $\beta$ to $\beta'$ and $\gamma$ to $\gamma'$, so no
carrier can preserve the equivalence relation.  Thus, $\sarr{X}{Y}$
has no carriers and is not saturated.  The (non-)proof given in the
paper misses this problem because it neglects to reason carefully
about which functions preserve the equivalence relation.

We could fix the example given here by demanding that if two carriers
of elements of a $\EE$-set realize the same expressions, they must be
related by $\eqsmall{i}{\EA}{}{}$.  However, it would still be
possible to construct a similar counter example by picking a smaller
set of realizers for $\gamma'$.  

One can imagine more complicated restrictions on $\Lambda$-sets and
$\EE$-sets which restore this property, but it is not clear how they
would influence the rest of the proof.  It is also possible that all
the specific uses of the product construction later in the paper
result in saturated $\Lambda$-sets.  However, because this regularity
property is implicitly relied on in countless places, tracking it
completely is beyond the scope of this survey.

\subsection{Models}
\label{secPTS:conditions}

We will now describe the model into which we interpret a labeled PTS.
This will be followed by four conditions, parameterized by the sets
$\sorts$, $\axioms$ and $\rules$.  The main result of the paper is
that when there exists a model satisfying the four conditions
instantiated with parameters corresponding to a particular PTS,
that system is strongly normalizing.

The four conditions are somewhat involved.  However, when using a
model where the set $\EE$ is empty, conditions 2-4 are trivially
satisfied.  This is the case for System F, so we recommend beginning
by understanding condition 1 and the model of System F in the next
section.  Then Section~\ref{secPTS:notCC} explains why this
construction does not suffice for CC, which may help motivate the
remaining conditions.

For each sort $s \in \sorts$, a model consists of
\begin{itemize}
\item an $\EE$-set $\EAup{s}$,
\item a saturated $\Lambda$-set $\EAdn{s}$, and
\item a bijection $\lift{s} : \EAdn{s}_0 \to \EAup{s}$.
\end{itemize}
When $\Gamma  \vdash  \ottnt{A}  \ottsym{:}  \ottnt{s}$, we intend for $\EAup{s}$ to contain a model of $A$
as a ``type'' and for $\EAdn{s}$ to contain a model of $A$ as a
``term''.  The lifting function $\lift{s}$ relates these two
interpretations: with each carrier of $\EAdn{s}$ we associate the
$\Lambda$-set that models its realizers as a type.  We will also use
the inverse of this function, which we write $\drop{s} : \EAup{s} \to
\EAdn{s}_0$.

The interpretation in Section~\ref{secPTS:interp} comes in two
corresponding levels.  Each type classified by a sort $s$ has a {\it
  type} interpretation as an element of $\EAup{s}$.  Every well-typed
expression also has a {\it term} interpretation as a carrier of the
$\Lambda$-set associated with its type.  Finally, each expression will
realize its term interpretation.  By condition 1.1 below, this will
imply the expression is strongly normalizing.

\subsubsection*{Condition 1: Uniformity of the universe hierarchy}

The following three properties ensure that the $\EE$-sets and
$\Lambda$-sets corresponding to sorts have a regular internal
structure and that there are relationships among them corresponding to
the sets $\axioms$ and $\rules$.

\begin{itemize}
\item[(1.1)] For each sort $s$, the elements of $\EAup{s}$ are
  saturated $\Lambda$-sets and the carriers of $\EAdn{s}$ are
  each realized by every strongly normalizing expression.
\item[(1.2)] If $(s_1,s_2) \in \axioms$, then $\EAdn{s_1} \in
  \EAup{s_2}$.
\item[(1.3)] Suppose $(s_1,s_2,s_3) \in \rules$ and that we have
  $\Lambda$-sets $X \in \EAup{s_1}$ and $Y_{\alpha} \in \EAup{s_2}$
  for each $\alpha \in X$.  If:
  $$
  \forall \alpha,\alpha' \lin X, \forall i \in \EE,\;
                 \eqsmall{i}{\EAup{s_1}}{\alpha}{\alpha'} 
    \Rightarrow  \eqbig{i}{\EAup{s_2}}{Y_{\alpha}}{Y_{\alpha'}}
  $$
  Then there is a $\Lambda$-set $\sarrdn{X}{Y} \in \EAup{s_3}$ and
  a $\Lambda$-iso $\isodn{X}{Y} : \sarr{X}{Y} \to \sarrdn{X}{Y}$.
\end{itemize}

While the first two sub-conditions are straightforward, some intuition
is helpful for the third.  If $ ( \ottnt{s_{{\mathrm{1}}}} , \ottnt{s_{{\mathrm{2}}}} , \ottnt{s_{{\mathrm{3}}}} ) \in \mathcal{R} $, then the PTS
allows function types whose domain is classified by $\ottnt{s_{{\mathrm{1}}}}$ and range
by $\ottnt{s_{{\mathrm{2}}}}$.  These function types can themselves be classified by
sort $\ottnt{s_{{\mathrm{3}}}}$.  Correspondingly, the condition says that when we may
form a $\Lambda$-set product in the model from $\EAup{s_1}$ to
$\EAup{s_2}$, there must be an isomorphic $\Lambda$-set in
$\EAup{s_3}$.  As we will see in Section~\ref{secPTS:notCC}, the
equivalence relation condition on the $\Lambda$-set product restricts
our attention to certain well-formed constructions to cope with the
size mismatch between set-theoretic function spaces and PTS functions.

The remaining three conditions impose regularity constraints on the
equivalence relations.  These are essentially just sanity checks, and
because we are not covering the details of the proofs we will not
fully discuss how they are used.  Roughly, conditions 2.1 and 3 check
that when two $\Lambda$-sets or carriers are related by
$\eqbig{i}{\EA}{}{}$ or $\eqsmall{i}{\EA}{}{}$, applying various
operations to them preserves the relation.  Conditions 2.2 and 4 check
that when carriers appear in the model of more than one sort, they are
treated uniformly.

\subsubsection*{Condition 2: Uniformity of $\eqbig{i}{\EA}{}{}$ and
  $\eqsmall{i}{\EA}{}{}$}

\begin{itemize}
\item[(2.1)] Suppose $X,X' \in \EAup{s_1}$ and $\EAdn{s_1} \in
  \EAup{s_2}$.  Then, for any $i \in \EE$:
  $$
  \eqbig{i}{\EAup{s_1}}{X}{X'} 
  \hspace{2em}\text{iff}\hspace{2em}
  \eqsmall{i}{\EAup{s_2}}{\drop{s_1}(X)}{\drop{s_1}(X')}
  $$
\item[(2.2)] Suppose $X_1, X_1' \in \EAup{s_1}$ and $X_2, X_2' \in
  \EAup{s_2}$ such that $\alpha$ is a carrier of both $X_1$ and $X_2$
  and $\alpha'$ is a carrier of both $X_1'$ and $X_2'$.  Then, for any
  $i \in \EE$:
  $$
  \eqsmall{i}{\EAup{s_1}}{\alpha}{\alpha'} 
  \hspace{2em}\text{iff}\hspace{2em}
  \eqsmall{i}{\EAup{s_2}}{\alpha}{\alpha'}
  $$
\end{itemize}

\subsubsection*{Condition 3: Uniformity of $\sarrdn{X}{Y}$ and
  $\isodn{X}{Y}$}
Suppose $ ( \ottnt{s_{{\mathrm{1}}}} , \ottnt{s_{{\mathrm{2}}}} , \ottnt{s_{{\mathrm{3}}}} ) \in \mathcal{R} $.  Let $\Lambda$-sets $X,X' \in
\EAup{s_1}$ and families of $\Lambda$-sets $Y_{\alpha},Y'_{\alpha'} \in
\EAup{s_2}$ for each $(\alpha,\alpha') \in X_0 \times X'_0$ be given
such that both families satisfy the hypothesis of (1.3).  For each
$i \in \EE$:
\begin{itemize}
\item[(3.1)] If 
  $\eqbig{i}{\sarr{\EAup{s_1}}{\EAup{s_2}}}{\sarr{X}{Y}}{\sarr{X'}{Y'}}$
  then
  $\eqbig{i}{\EAup{s_3}}{\sarrdn{X}{Y}}{\sarrdn{X'}{Y'}}$.
\item[(3.2)] Suppose
  $\eqbig{i}{\sarr{\EAup{s_1}}{\EAup{s_2}}}{\sarr{X}{Y}}{\sarr{X'}{Y'}}$.
  Then for any $f \lin \sarr{X}{Y}$ and $g \lin \sarr{X'}{Y'}$:
  $$
  \eqsmall{i}{\sarr{\EAup{s_1}}{\EAup{s_2}}}{f}{g}
  \hspace{2em}\text{iff}\hspace{2em}
  \eqsmall{i}{\EAup{s_3}}{\isodn{X}{Y}(f)}{\isodn{X'}{Y'}(g)}
  $$
\end{itemize}

\subsubsection*{Condition 4: Uniformity of $\lift{s}$ and $\drop{s}$}
\begin{itemize}
\item[(4.1)] If $X \in \EAup{s_1}$ and $X \in \EAup{s_2}$ then
  $\drop{s_1}(X) = \drop{s_2}(X)$.
\item[(4.2)] If $\alpha \lin \EAdn{s_1}$ and $\alpha \lin \EAdn{s_2}$
  then $\lift{s_1}(\alpha) = \lift{s_2}(\alpha)$.
\end{itemize}
These two conditions indicate that the sort subscripts on the lifting
operation and its inverse are only annotations; they do not influence
the behavior of the functions.

\subsection{A model for System F}
\label{secPTS:systemF}

We will use two simple $\Lambda$-set constructions in the model for
System F.
\begin{definition}
  A $\Lambda$-set $X$ is \textit{degenerate} if $X_0 = \{S\}$ where
  $S$ is a saturated set and $a \models_{X} S$ iff $a \in S$.
  We refer to $S$ as the \textit{underlying set} of $X$, and write
  $\DG$ for the set of all degenerate $\Lambda$-sets.
\end{definition}

\begin{definition}
  When $X$ is any non-empty set, we define an associated $\Lambda$-set
  $J(X)$ whose carrier set is $X$.  Each element of $X$ is realized by
  every strongly normalizing expression.
\end{definition}

Recall that System F is given by:
$$
\sorts = \{\sortstar,\sortbox\} \hspace{4em}
\axioms = \{(\sortstar,\sortbox)\} \hspace{4em}
\rules = \{(\sortstar,\sortstar,\sortstar),(\sortbox,\sortstar,\sortstar)\}
$$
In the model of System F, we pick $\EE = \emptyset$.  So, we will not
have to define any of the $\EE$-set relations.  We pick
$$
\EAup{\sortstar} = \DG  \hspace{6em}  \EAup{\sortbox} = \{ J(\DG) \}
$$
and, for each sort $s$, set $\EAdn{s} = J(\EAup{s})$.  The bijection
$\lift{s} : \EAdn{s}_0 \to \EAup{s}$ is then simply the identity
function.

We must verify that this model has the appropriate properties.  Of
these, conditions 2-4 are vacuous because $\EE$ is empty.  Conditions
1.1 and 1.2 are apparent.  It only remains to check 1.3.

\begin{proof}
  Suppose that $ ( \ottnt{s_{{\mathrm{1}}}} , \ottnt{s_{{\mathrm{2}}}} , \ottnt{s_{{\mathrm{3}}}} ) \in \mathcal{R} $ and we have $\Lambda$-sets $X
  \in \EAup{s_1}$ and $Y_{\alpha} \in \EAup{s_2}$ for each $\alpha
  \lin X$.  Because $s_2 = s_3 = \sortstar$, each $Y_{\alpha}$ is a
  degenerate $\Lambda$-set.

  The carriers of $\sarr{X}{Y}$ are the functions which map each
  $\alpha \lin X$ to a carrier of $Y_{\alpha}$.  But because each
  $Y_{\alpha}$ has just one carrier, there is only one such function
  $f$.

  We must pick a degenerate $\Lambda$-set for $\sarrdn{X}{Y}$, so
  pick the one whose underlying set is the realizers of $f$ in
  $\sarr{X}{Y}$.\footnote{Here we are implicitly relying on the
    problematic lemma from Section~\ref{secPTS:problem}.  However,
    expanding the definitions when $s_1 = \sortstar$ or $s_1 =
    \sortbox$ reveals that the constructions remain saturated in this
    particular case.}  Then the $\Lambda$-iso $\isodn{X}{Y}$ simply
  maps the only carrier of $\sarr{X}{Y}$ to the only carrier of
  $\sarrdn{X}{Y}$ and trivially satisfies the condition that realizers
  are preserved.
\end{proof}

\subsection{Does this model work for CC?}
\label{secPTS:notCC}

It is instructive to consider why the model selected for System F is
not sufficient for CC.  This will motivate the use of $\EE$-sets and
their equivalence relations.  Recall that CC is the PTS given by:
$$
\sorts = \{\sortstar,\sortbox\} \hspace{4em}
\axioms = \{(\sortstar,\sortbox)\} \hspace{4em}
\rules = \{  (\sortstar,\sortstar,\sortstar),
           \,(\sortstar,\sortbox,\sortbox),
           \,(\sortbox,\sortstar,\sortstar),
           \,(\sortbox,\sortbox,\sortbox)\}
$$
The only difference from System F is the addition of two new rules.
We find the problem in satisfying property 1.3 for the rule
$(\sortbox,\sortbox,\sortbox)$.

Suppose $X \in \EAup{\sortbox}$ and $Y_{\alpha} \in \EAup{\sortbox}$
for each $\alpha \lin X$.  Referring back to the definition of
$\EAup{\sortbox}$ for our model, we see this means $X$ and the
$Y_{\alpha}$s are all $J(\DG)$.

Condition 1.3 demands that we find an element of $\EAup{\sortbox}$
which is isomorphic to $\sarr{X}{Y}$.  The only such element is
$J(\DG)$, so we are being asked to show $J(\DG)$ is isomorphic to
$\sarr{X}{Y}$.  But the carrier set of $J(\DG)$ is $\DG$, and the
carrier set of $\sarr{X}{Y}$ is all the set theoretic functions from
$\DG$ to $\DG$.  Plainly, there is no bijection between these sets.

This problem is fundamental, since the axiom $(\sortstar,\sortbox)$
means $\EAup{\sortbox}$ must contain $\EAdn{\sortstar}$.  To resolve
it, we use a non-empty $\EE$.  The equivalence relations will break
our model of $\sortbox$ up into an infinite hierarchy of levels.  In
situations like this where self-reference and set-theoretic size
constraints are a problem, we will simply move to a ``higher'' level.

\subsection{A model for CC}
\label{secPTS:CC}

We pick $\EAup{\sortstar} = \DG$ and $\EAdn{s} = J(\EAup{s})$ for each
$s$ as we did in the model of System F.  Correspondingly, each
$\lift{s}$ is the identity function.  The definition of
$\EAup{\sortbox}$ is broken into levels.  Define:
\begin{itemize}
\item $\lev_1 = \{J(\DG)\}$
\item For each $n \in \mathbb{N}^{+}$, 
  $$
  \lev_{\le n} = \displaystyle\bigcup_{1 \le k \le n} \lev_{k}
  $$
\item For each $n \in \mathbb{N}^{+}$
  \begin{center}
    \begin{tabular}{ccll}
      $\lev_{n+1}$ 
         & = 
         & $\{\sarr{X}{Y}\;|\; X \in \EAup{\sortstar},\,$&
                               $Y_{\alpha} \in \lev_{n} 
                                 \text{ for each } \alpha \lin X \}$\\
         & $\cup$
         & $\{\sarr{X}{Y}\;|\; X \in \lev_{n},\,$&
                               $Y_{\alpha} \in \lev_{\le n} 
                                 \text{ for each } \alpha \lin X \}$\\
         & $\cup$
         & $\{\sarr{X}{Y}\;|\; X \in \lev_{\le n},\,$&
                               $Y_{\alpha} \in \lev_{n} 
                                 \text{ for each } \alpha \lin X \}$
    \end{tabular}
  \end{center}
\end{itemize}

Then $\EAup{\sortbox}$ is the collection of all these levels:
$$
\EAup{\sortbox} = \displaystyle\bigcup_{n \in \mathbb{N}^{+}} \lev_{n}
$$
Notice that all the levels are disjoint.  Each level collects the
$\Lambda$-set products whose domain or range is formed from elements
of the previous level.

Finally we define the equivalence relations for our $\EE$-sets.  In
this model, we make $\EE$ a singleton set $\{1\}$, so there are two
relations to define for each sort.
\begin{itemize}
\item $\eqbig{1}{\EAup{\sortstar}}{}{}$ relates every two elements of
  $\EAup{\sortstar}$.
\item $\eqsmall{1}{\EAup{\sortstar}}{}{}$ relates every two carriers
  of every two elements of $\EAup{\sortstar}$.
\item $\eqbig{1}{\EAup{\sortbox}}{}{}$ relates two elements of
  $\EAup{\sortbox}$ iff they are in the same $\lev_{n}$.
\item $\eqsmall{1}{\EAup{\sortbox}}{}{}$ relates every two carriers of
  every two elements of $\EAup{\sortbox}$.
\end{itemize}

These constructions satisfy conditions 1.1 and 1.2 as in the model for
System F.\footnote{This is another implicit use of the mistaken
  theorem discussed in Section~\ref{secPTS:problem}.  However, one can
  verify that $\EAup{\sortbox}$ here does contain only saturated
  $\Lambda$-sets.  The situation is less clear in the paper's model of
  ECC, where the relations $\eqsmall{i}{\EAup{s}}{}{}$ are more
  complicated.}  The same proof of condition 1.3 for rules
$(\sortstar,\sortstar,\sortstar)$ and $(\sortbox,\sortstar,\sortstar)$
also applies.  We now consider how this model solves the problem we
encountered with rules $(\sortstar,\sortbox,\sortbox)$ and
$(\sortbox,\sortbox,\sortbox)$ in the previous section.  
\begin{proof}
  Suppose $X \in \EAup{s}$ for some $s$ and $Y_{\alpha} \in
  \EAup{\sortbox}$ for each $\alpha \lin X$.  Suppose also that:
$$
\forall \alpha,\alpha' \lin X,
\eqsmall{1}{\EAup{s_1}}{\alpha}{\alpha'} \;\Rightarrow\;
\eqbig{1}{\EAup{\sortbox}}{Y_{\alpha}}{Y_{\alpha'}}
$$
Because the relations $\eqsmall{1}{\EAup{s}}{}{}$ equate any two
carriers, this implies that all of the $Y_{\alpha}$s are related by
$\eqbig{1}{\EAup{\sortbox}}{}{}$.  This means they are all in the same
set $\lev_n$ for some $n$.  Define the natural number $m$ to be $1$ if
$s = \sortstar$ and otherwise to be the $m$ such that $X \in \lev_m$.
Then $\sarr{X}{Y} \in \lev_{\max(n,m)+1}$.  So we pick $\sarr{X}{Y}$
itself for $\sarrdn{X}{Y}$ and the identity function for
$\isodn{X}{Y}$.
\end{proof}

We check each of the remaining conditions.  Of these, 3.1 and 3.2 are
somewhat involved and the rest are easy.
\begin{itemize}
\item[(2.1)] The only case where $\EAdn{s_1} \in \EAup{s_2}$ is when
  $s_1$ is $\sortstar$ and $s_2$ is $\sortbox$.  The required condition
  follows easily since $\eqbig{1}{\EAup{s_1}}{}{}$ relates every two
  elements of $\EAup{s_1}$ and $\eqsmall{1}{\EAup{s_2}}{}{}$ relates
  every two carriers of elements of $\EAup{s_2}$.

\item[(2.2)] The carriers of elements of $\EAup{\sortstar}$ are all
  expressions.  The carriers of elements of $\EAup{\sortbox}$ are all
  $\Lambda$-sets.  Thus we must have $s_1 = s_2$ and the condition is
  trivial.

\item[(3.1)] When $s_2 = s_3 = \sortstar$ this condition is trivial
  because $\eqbig{i}{\EAup{\sortstar}}{}{}$ relates every two elements
  of $\EAup{\sortstar}$.  Otherwise suppose the $\Lambda$-sets $X,X'
  \in \EAup{s_1}$ and the families of $\Lambda$-sets
  $Y_{\alpha},Y'_{\alpha'} \in \EAup{\sortbox}$ are given as
  specified, so that:
  $$
  \eqbig{1}{\sarr{\EAup{s_1}}{\EAup{\sortbox}}}{\sarr{X}{Y}}{\sarr{X'}{Y'}}
  $$
  We must show that $\sarrdn{X}{Y}$ and $\sarrdn{X'}{Y'}$ occupy the
  same $\lev_n$.  Recall that $\sarrdn{X}{Y} = \sarr{X}{Y}$ and
  $\sarrdn{X'}{Y'} = \sarr{X'}{Y'}$.  The result then follows from the
  definition of $\eqbig{1}{\sarr{\EAup{s_1}}{\EAup{s_2}}}{}{}$, which
  checks that $X$ and $X'$ are on the same level and each pair
  $Y_{\alpha}$ and $Y_{\alpha'}$ are on the same level.

\item[(3.2)] Suppose $ ( \ottnt{s_{{\mathrm{1}}}} , \ottnt{s_{{\mathrm{2}}}} , \ottnt{s_{{\mathrm{3}}}} ) \in \mathcal{R} $ and
  $\eqbig{1}{\sarr{\EAup{s_1}}{\EAup{s_2}}}{\sarr{X}{Y}}{\sarr{X'}{Y'}}$.
  Let $f \lin \sarr{X}{Y}$ and $g \lin \sarr{X'}{Y'}$ be given.  Since
  $\eqsmall{1}{\EAup{s_3}}{}{}$ always relates any two carriers of
  $\EAup{s_3}$, we must show
  $\eqsmall{1}{\sarr{\EAup{s_1}}{\EAup{s_2}}}{f}{g}$.

  This follows immediately by the definition of
  $\eqsmall{1}{\sarr{\EAup{s_1}}{\EAup{s_2}}}{}{}$ and the
  observation that $\eqsmall{1}{\EAup{s_2}}{}{}$ also relates any two
  carriers of $\EAup{s_2}$ (and thus any two elements in the in the
  range of $f$ and $g$).

\item[(4)] The models of the sorts are disjoint, so this condition is
  trivial.
  
\end{itemize}

\subsection{The interpretation}
\label{secPTS:interp}

The interpretations of types and terms should not be very surprising.
As we have mentioned, the interpretation of a type will be a
$\Lambda$-set and the interpretation of a term will be a carrier.  The
soundness theorem will say that every well-typed expression realizes its
term interpretation.

We define three functions by mutual recursion: $\ptsCtxI{\Gamma}$,
$\ptsTypI{\Gamma}{a}$ and $\ptsTrmI{\Gamma}{a}$.  The soundness
theorem will show that the first is defined whenever $\vdash_{t}  \Gamma$ and
the latter two whenever $\Gamma  \vdash_{t}  a  \ottsym{:}  A$ and $\gamma \in
\ptsCtxI{\Gamma}$.  In general $\gamma$ will be an $n$-tuple of pairs
associating variables from the context with carriers.  We write
$\gamma(x)$ for the carrier associated with $x$, when it exists.

\begin{align}
  \ptsCtxI{ \cdot } & := \emptyset \\
  \ptsCtxI{ \Gamma , \ottmv{x} \!:\! A } & := 
      \{ (\gamma,(x,\alpha))\;|\; 
                \gamma \in \ptsCtxI{\Gamma} 
         \wedge \alpha \lin \ptsTypI{\Gamma}{A} \}\\
  \notag \\
  \ptsTypI{\Gamma}{\ottnt{s}} &:= \EAdn{s}
     \label{eqs:ptsTypSort}\\
  \ptsTypI{\Gamma}{\ottsym{(}  \ottmv{x}  \ottsym{:}  A  \ottsym{)}  \to  B} &:= 
    \sarrdn{\ptsTypI{\Gamma}{A}}
           {\ptsTypI[\gamma,\_]{ \Gamma , \ottmv{x} \!:\! A }{B}}
     \label{eqs:ptsTypPi}\\
  \ptsTypI{\Gamma}{\ottnt{A}} &:= \lift{s}(\ptsTrmI{\Gamma}{A})
     \label{eqs:ptsTypDefault}\\
  \notag\\
  \ptsTrmI{\Gamma}{x} &:= \gamma(x)
     \label{eqs:ptsTrmVar}\\
  \ptsTrmI{\Gamma}{ \lambda _{( \ottmv{x} : A ) \to B }. b } &:= 
    \isodn{\ptsTypI{\Gamma}{A}}
          {\ptsTypI[\gamma,\_]{ \Gamma , \ottmv{x} \!:\! A }{B}} 
       (\ptsTrmI[\gamma,\_]{ \Gamma , \ottmv{x} \!:\! A }{b})
     \label{eqs:ptsTrmLam}\\
  \ptsTrmI{\Gamma}{ \mathsf{app} _{( \ottmv{x} : A ) \to B } ( a , b ) } &:=
    (\isoup{\ptsTypI{\Gamma}{A}}
           {\ptsTypI[\gamma,\_]{ \Gamma , \ottmv{x} \!:\! A }{B}}
       (\ptsTrmI{\Gamma}{a}))\;(\ptsTrmI{\Gamma}{b})
    \label{eqs:ptsTrmApp}\\
  \ptsTrmI{\Gamma}{a} &:=
    \drop{s}(\ptsTypI{\Gamma}{a})
    \label{eqs:ptsTrmDefault}
\end{align}

There are a few notational infelicities to explain.  First,
equations~\ref{eqs:ptsTypDefault} and~\ref{eqs:ptsTrmDefault} are
meant to apply only when the previous clauses in the respective
definitions do not.  The remark after the definition of condition 4 in
Section~\ref{secPTS:conditions} justifies the use of an arbitrary $s$
in these cases.  By $\isoup{\ptsTypI{\Gamma}{A}}
{\ptsTypI[\gamma,\_]{ \Gamma , \ottmv{x} \!:\! A }{B}}$ in
equation~\ref{eqs:ptsTrmApp} we mean the inverse of
$\isodn{\ptsTypI{\Gamma}{A}}
{\ptsTypI[\gamma,\_]{ \Gamma , \ottmv{x} \!:\! A }{B}}$.  Finally, throughout the
definition we wrote $\ptsTypI[\gamma,\_]{ \Gamma , \ottmv{x} \!:\! A }{B}$ and
$\ptsTrmI[\gamma,\_]{ \Gamma , \ottmv{x} \!:\! A }{b}$ for the functions
$$
  \alpha \lin \ptsTypI{\Gamma}{A} \mapsto
              \ptsTypI[\gamma,(x,\alpha)]{ \Gamma , \ottmv{x} \!:\! A }{B}
$$
and
$$
  \alpha \lin \ptsTypI{\Gamma}{A} \mapsto
              \ptsTrmI[\gamma,(x,\alpha)]{ \Gamma , \ottmv{x} \!:\! A }{b}
$$
respectively.  The first can be viewed as an indexed family of
$\Lambda$-sets suitable for use as the second argument to a
$\Lambda$-set product construction.  The second is a carrier of
products formed with the first.

This definition is more intimidating than the interpretation from
the previous proof, but none of it is surprising.  More, it has the
advantage that the connection between the interpretations of types and
terms is explicit as part of their definition.  We describe each
clause to help with the symbolic burdern:
\begin{itemize}
\item The definition of $\ptsCtxI{\Gamma}$ says that when $\gamma \in
  \ptsCtxI{\Gamma}$, each $x$ in $\Gamma$ should be paired with a
  carrier of the interpretation of its given type.  This is
  essentially the definition of the $\sigma \models \Gamma$ judgement
  from the last section.

\item The definition of $\ptsTypI{\Gamma}{\ottnt{A}}$ associates with each
  well-typed expression a $\Lambda$-set in the model (assuming the
  four conditions are met).  Clause~\ref{eqs:ptsTypSort} says that
  when $A$ is a sort $s$, we pick the $\Lambda$-set explicitly
  specified by the model ($\EAdn{s}$).

  In clause~\ref{eqs:ptsTypPi}, we consider a function type $\ottsym{(}  \ottmv{x}  \ottsym{:}  A  \ottsym{)}  \to  B$.  The typing rule \textsc{TTPi} says this will be a
  valid function type in sort $s_3$ when $A$ is in sort $s_1$,
  $B$ is in sort $s_2$ and $ ( \ottnt{s_{{\mathrm{1}}}} , \ottnt{s_{{\mathrm{2}}}} , \ottnt{s_{{\mathrm{3}}}} ) \in \mathcal{R} $.  The
  $\Lambda$-set $\sarr{\ptsTypI{\Gamma}{A}}
  {\ptsTypI[\gamma,\_]{ \Gamma , \ottmv{x} \!:\! A }{B}}$ models the functions
  from the interpretation of $A$ to the interpretation of
  $B$, but we do not know if it exists anywhere in our model.
  Luckily, condition 1.3 guarantees that an isomorphic $\Lambda$-set
  lives in $\EAup{s_3}$.  We pick that one, since this is just where
  the interpretation of the function type belongs.
  
  The catch-all clause~\ref{eqs:ptsTypDefault} handles expressions
  that look more like terms than types.  In that case, we use the term
  interpretation and then lift the resulting carrier to a
  $\Lambda$-set with the provided $\lift{}$ function.  For example, if
  $A$ is the application of a type function $\ottnt{a} \, \ottnt{b}$, the term
  interpretation will interpret $\ottnt{a}$ and apply the result to the
  term interpretation of $\ottnt{b}$.  There is a similar clause in the
  term interpretation that will call back to the type interpretation
  when it reaches subcomponents that looked more like types.

\item The definition of $\ptsTrmI{\Gamma}{\ottnt{a}}$ associates the
  carrier of a $\Lambda$-set in the model with each well-typed
  expression.  Here, clause~\ref{eqs:ptsTrmVar} handles variables
  using the provided environment.  The final default
  clause~\ref{eqs:ptsTrmDefault} mirrors
  clause~\ref{eqs:ptsTypDefault} from the previous definition.

  Functions $ \lambda _{( \ottmv{x} : A ) \to B }. b $ are interpreted by
  clause~\ref{eqs:ptsTrmLam}.  Here,
  $\ptsTrmI[\gamma,\_]{ \Gamma , \ottmv{x} \!:\! A }{b}$ is a function from
  carriers in the type interpretation of $A$ to carriers in the
  type interpretation of $B$.  Such functions are themselves
  carriers of $\sarr{\ptsTypI{\Gamma}{A}}
  {\ptsTypI[\gamma,\_]{ \Gamma , \ottmv{x} \!:\! A }{B}}$.  As we saw in the
  discussion of clause~\ref{eqs:ptsTypPi}, this $\Lambda$-set may not
  be in our model.  So we use the function
  $\isodn{\ptsTypI{\Gamma}{A}}
  {\ptsTypI[\gamma,\_]{ \Gamma , \ottmv{x} \!:\! A }{B}}$, which condition 1.3
  guarantees will return a corresponding carrier in the model.
  
  The case for applications is similar.  The soundness theorem will
  show that, if the application $ \mathsf{app} _{( \ottmv{x} : A ) \to B } ( a , b ) $ is
  well-typed, the interpretation of $a$ will be a carrier of
  $\sarrdn{\ptsTypI{\Gamma}{\ottnt{A}}}{\ptsTypI[\gamma,\_]{ \Gamma , \ottmv{x} \!:\! A }{B}}$.
  We use the provided $\Lambda$-iso to convert this to a carrier of
  $\sarr{\ptsTypI{\Gamma}{\ottnt{A}}}{\ptsTypI[\gamma,\_]{ \Gamma , \ottmv{x} \!:\! A }{B}}$
  so that applying it to the interpretation of $b$ will yield a
  carrier in the type interpretation of $B$.
\end{itemize}

\subsection{Strong Normalization}
\label{secPTS:sn}

Strong normalization follows from two key theorems.  The first says
that when an expression is well-typed, its term interpretation is a
carrier of the interpretation of its type.  The second says that each
expression realizes its term interpretation.  Since condition 1.1
guarantees that all $\Lambda$-sets in the model are saturated, the
realizers are all strongly normalizing.

\begin{theorem}[Soundness]
  If $\vdash_{t}  \Gamma$, then $\ptsCtxI{\Gamma}$ is well-defined.  If $\Gamma  \vdash_{t}  a  \ottsym{:}  A$ and $\gamma \in \ptsCtxI{\Gamma}$ then:
  \begin{itemize}
  \item $\ptsTrmI{\Gamma}{a}$ and $\ptsTypI{\Gamma}{A}$ are
    well-defined, and $\ptsTrmI{\Gamma}{a} \lin
    \ptsTypI{\Gamma}{A}$.
  \item If $A = s$, then $\ptsTypI{\Gamma}{a}$ is well-defined
    and an element of $\EAup{s}$.
  \end{itemize}
\end{theorem}

\begin{theorem}[Self-realization]\label{thmPTS:selfrealize}
  Suppose $\Gamma  \vdash_{t}  a  \ottsym{:}  A$ and $\Gamma = x_1:B_1,\ldots,x_n:B_n$.
  Let $((x_1,\beta_1),\ldots,(x_n,\beta_n)) \in \ptsCtxI{\Gamma}$ and
  expressions $b_1,\ldots,b_n$ be given such that for each $1 \le i
  \le n$:
  $$
  b_i \models_{\ptsTypI[(x_1,\beta_1),\ldots,(x_{i-1},\beta_{i-1})]
              {x_1:B_1,\ldots,x_{i-1}:B_{i-1}}{B_i}} \beta_i
  $$
  Then:
  $$
  [b_1/x_1]\ldots[b_n/x_n] a 
     \models_{\ptsTypI{\Gamma}{A}} \ptsTrmI{\Gamma}{a}
  $$
\end{theorem}

Strong normalization then follows just as it did in the previous
proof.  In Theorem~\ref{thmPTS:selfrealize}, the type interpretation
of each $B_i$ has a carrier which is realized by any strongly
normalizing base expression, and thus any variable.  Pick that
realizer for $\beta_i$ and pick $x_i$ for $b_i$.  Then the
substitutions in the conclusion of the theorem do nothing and we find
that $a$ realizes its own term interpretation.  Since term
interpretations are carriers in saturated $\Lambda$-sets, $a$ is
strongly normalizing.

We do not go into the details of the proofs of these theorems.
Indeed, the proofs given seem inadequate.  For example, consider the
first lemma the paper gives after the interpretation:
\vspace{-0.8em}
\begin{quote}
  $$
  \ptsTrmI{\Gamma}{ \mathsf{app} _{( \ottmv{x} : A ) \to B } ( \ottsym{(}   \lambda _{( \ottmv{x} : A ) \to B }. a   \ottsym{)} , b ) } =
  \ptsTrmI[\gamma,(x,b)]{\Gamma}{a}
  $$
\end{quote}
The authors do not identify the hypotheses of this lemma.  However,
observe that the lambda term is interpreted as a function whose domain
is the carriers of $\ptsTypI{\Gamma}{A}$.  Thus, if $\ptsTrmI{\Gamma}{b}$ is not
such a carrier, the left-hand side will not be defined while the
right-hand side may be (for example, if $x$ doesn't occur in $a$).
The lemma is subsequently used in situations where we only know $\Gamma  \vdash_{t}  b  \ottsym{:}  A$, and the soundness theorem itself would be needed to
show this is enough.  A much more careful proof of the soundness
theorem is needed.

All that remains is to relate strong normalization for a labeled PTS
to strong normalization for an ordinary PTS.  To do this, define an
operation $|a|$ on labeled expressions which simply erases the extra
annotations to obtain an ordinary PTS expression.  Extend this
operation to contexts $|\Gamma|$ by applying it to each type.  We would
like to know that if $\Gamma  \vdash  \ottnt{a}  \ottsym{:}  \ottnt{b}$ in an ordinary PTS, then there
is a labeled context $\Gamma'$ and labeled expressions $a',b'$ such
that $|\Gamma'| = \Gamma$, $|a'| = a$, $|b'| = b$ and $\Gamma'  \vdash_{t}  a'  \ottsym{:}  b'$.
Since $a'$ is well-typed in the labeled PTS, all the type annotations
on beta redexes must match, and thus that its normalization behavior
corresponds to that of $a$.

The proof of this theorem is mostly straightforward by induction.  The
only problem comes in the case where the derivation used the
conversion rule.  Here we must show a relationship between labeled,
tight conversion and ordinary PTS conversion.  The paper shows that
the two agree in the case of well-typed expressions:
\begin{lemma}
  Suppose $\Gamma  \vdash_{t}  a  \ottsym{:}  A$ and $\Gamma  \vdash_{t}  b  \ottsym{:}  B$.  If $|a|
   =_{\beta}  |b|$ then $a  =_{\beta}^{t}  b$.
\end{lemma}
This is unsurprising, since in well-typed expressions, labeled or
not, beta reductions only occur when the function's domain type agrees
with the type of its argument.  The authors give a detailed proof.


\section{Discussion and Conclusion}
\label{sec:conclusion}

The original goal of this project was to survey several very different
strong normalization proofs for the calculus of constructions.  To
that end, we picked three attempts that target different structures
(sets of expressions, realizability semantics, and F$_{\omega}$).
Each paper's proof was targeted toward different additional goals.
The first considered extensions with various datatypes and recursion.
The second tried to model a large class of pure type systems.  The
last gave a relatively straightforward translation to simpler system,
demonstrating that the proof-theoretic complexity of CC's strong
normalization argument is no greater than that of F$_{\omega}$.

Despite this, the three proofs are remarkably similar.  Each gives a
type interpretation $ \fomTypI[ ]{ \ottnt{A} } $ followed by a term interpretation
$ \fomTrmI[ ]{ \ottnt{a} } $.  Then, when $\Gamma  \vdash  \ottnt{a}  \ottsym{:}  \ottnt{A}$, a simple relationship
between the two translations is demonstrated (for example, $ \fomTrmI[ ]{ \ottnt{a} } 
\in  \fomTypI[ ]{ \ottnt{A} } $).  Finally, this relationship is shown to imply that
the expression $a$ itself is strongly normalizing.  Though the models
targeted by these functions are different in each proof, they share a
considerable amount of structure.  For example, saturated sets of
expressions are very useful in both the first and second proofs.  They
are often used in proofs of strong normalization for F$_{\omega}$ as
well.

There are many commonalities even in the specifics of the
interpretations.  Compare the kind interpretation $V$ from
Section~\ref{secSAT:kinds} with $V$ in Section~\ref{secFOM:types}.
Though one targets collections saturated sets and the other
F$_{\omega}$, both cope with the kind structure of CC in the same way.
Moreover, examining this similarity yields a cleaner understanding of
both proofs.  In the context of saturated sets, it is tricky to
motivate the definition of the kind interpretation $V(\ottsym{(}  \ottmv{x}  \ottsym{:}  \ottnt{A}  \ottsym{)}  \to  \ottnt{B})$
when $A$ is a type.  For example, we might plausibly have picked
functions from values or expressions to $V(B)$ instead of just $V(B)$
itself.  But in the translation to F$_{\omega}$ we find a snappy
explanation: this definition simply erases dependency.  By unifying
the presentations and providing a common narrative through the three
proofs, we found additional clarity in each.

The papers have differences, too.  These typically result from the
motivations of the authors.  For example, the translation to
F$_{\omega}$ is certainly the simplest of the developments and
provides the most confidence in the result.  It achieves this by
relying on the existing strong normalization result for the simpler
system.  On the downside, the authors do not consider extensions.

The most complicated proof uses realizability semantics.  The
structure of the $\EE$-sets is somewhat intimidating, and in the end
the model of CC uses relatively little of its expressiveness.  The
authors aim to provide a technique which extends to ECC, but perhaps
in this they overreach: the reasoning about the structures involved is
sometimes mistaken, and it is not clear how simple it would be to
repair.  

The traditional approach using saturated sets falls somewhere between
these two.  Though this technique must be extended to cope with CC,
the proof is somewhat familiar and does not require any new
structures.  The authors succeed in extending the approach to various
common datatypes with recursion, but do not consider more complicated
additions like a predicative hierarchy of universes or large
eliminations.

The popularity of dependently-typed programming languages continues to
grow.  So, too, do their lists of features.  Our understanding of
their metatheory and of strong normalization in particular has not
quite kept pace.  Thus, while the strong normalization of CC has been
considered a settled issue for more than two decades, understanding
its fundamentals is more important now than ever.

\bibliographystyle{abbrvnat}
\bibliography{refs}

\appendix
\section{Details for Section~\ref{sec:flexible}}

\begin{theorem*}[Soundness of the interpretation]
  Suppose $\Gamma  \vdash  \ottnt{a}  \ottsym{:}  \ottnt{A}$ and $\sigma \models \rho : \Gamma$.  Then
  $\subst{a} \in \interp{A}$.
\end{theorem*}

\begin{proof}
  We go by induction on the typing derivation $\D$
  \begin{itemize}
  \item[{\bf Case:}] $\D = \lowered{
      \infer[TSort]
        {\vdash  \Gamma\\  ( \ottnt{s_{{\mathrm{1}}}} , \ottnt{s_{{\mathrm{2}}}} ) \in \mathcal{A} }
        {\Gamma  \vdash   \ast   \ottsym{:}   \Box }
    }$

    The only axiom is $\ottsym{(}   \ast   \ottsym{,}   \Box   \ottsym{)}$.  This case is immediate, as
    $\subst{ \ast } =  \ast  \in \SN = \interp{ \Box }$

  \item[{\bf Case:}] $\D = \lowered{
      \infer[TVar]
      {\vdash  \Gamma \\  ( \ottmv{x} : \ottnt{A} ) \in  \Gamma }
      {\Gamma  \vdash  \ottmv{x}  \ottsym{:}  \ottnt{A}}
    }$

    By the assumption $\sigma \models \rho : \Gamma$, the rule's
    second premise implies that $\subst{x} \in \interp{A}$.

  \item[{\bf Case:}] $\D = \lowered{
      \infer[TPi]
        {\infer{}{\D_1 \\\\ \Gamma  \vdash  \ottnt{A}  \ottsym{:}  \ottnt{s_{{\mathrm{1}}}}} \\ 
         \infer{}{\D_2 \\\\  \Gamma , \ottmv{x} \!:\! \ottnt{A}   \vdash  \ottnt{B}  \ottsym{:}  \ottnt{s_{{\mathrm{2}}}}} \\
          ( \ottnt{s_{{\mathrm{1}}}} , \ottnt{s_{{\mathrm{2}}}} , \ottnt{s_{{\mathrm{3}}}} ) \in \mathcal{R} }
        {\Gamma  \vdash  \ottsym{(}  \ottmv{x}  \ottsym{:}  \ottnt{A}  \ottsym{)}  \to  \ottnt{B}  \ottsym{:}  \ottnt{s_{{\mathrm{3}}}}}
    }$

    We must show that $(x : \subst{A}) \to \subst{B} \in \SN$.
    By IH for $\D_1$, $\subst{A} \in \SN$, so it only remains to show
    that $\subst{B} \in \SN$.

    The IH for $\D_2$ says that given any $\sigma'$ and $\rho'$ such
    that $\sigma' \models \rho' :  \Gamma , \ottmv{x} \!:\! \ottnt{A} $, we have $\rho' B \in
    SN$.  For $\sigma'$, pick $\sigma$ if $s_1$ is $\ast$, or extend
    it with a canonical inhabitant of $V(A)$ if $s_1$ is $ \Box $.
    For $\rho'$, pick $\rho[x \mapsto x]$.  Then $\rho' x = x$, which
    is in $\interp{A}$ since the latter must be a saturated set by
    lemma~\ref{lem:interptype}.  We conclude that $\subst{B} =
    \subst[\rho']{B} \in \SN$ as desired.

  \item[{\bf Case:}] $\D = \lowered{
      \infer[TLam]
        {\infer{}{\D_1 \\\\ \Gamma  \vdash  \ottnt{A}  \ottsym{:}  \ottnt{s}} \\
         \infer{}{\D_2 \\\\  \Gamma , \ottmv{x} \!:\! \ottnt{A}   \vdash  \ottnt{b}  \ottsym{:}  \ottnt{B}}}
       {\Gamma  \vdash  \lambda  \ottmv{x}  \ottsym{:}  \ottnt{A}  \ottsym{.}  \ottnt{b}  \ottsym{:}  \ottsym{(}  \ottmv{x}  \ottsym{:}  \ottnt{A}  \ottsym{)}  \to  \ottnt{B}}
    }$

    There are two subcases: $s$ is either $\ast$ or $\Box$.  In either
    case, the IH for $\D_1$ gives us that $\subst{A} \in \SN$.  Since
    $x$ is a bound variable, we may pick it to be fresh for the domain
    and range of $\rho$.
    \begin{itemize}
    \item Suppose $s$ is $\ast$.  Then we must show $\lambda
      x:\subst{A} . \subst{b} \in \sarr{\interp{A}}{\interp{B}}$.  So
      let $a \in \interp{A}$ be given, and observe it is enough to
      show $(\lambda x:\subst{A} . \subst{b})\;a \in \interp{B}$.

      We have $\sigma \models \rho[x \mapsto a] : G,x:A$.  Thus, the
      IH for $\D_2$ gives us $\subst[\rho[x \mapsto a\rbracket ]{b} \in
      \interp{B}$.  But we also know:
      $$
        (\lambda x:\subst{A} . \subst{b})\;a  \leadsto 
        [a/x]\subst{b} =
        \subst[ \rho[x \mapsto a\rbracket ]{b}
      $$
      And this step contracts a key redex.  So, it suffices to show
      that $(\lambda x:\subst{A} . \subst{b})\;a \in \SN$.
      
      Lemma~\ref{lem:interptype} gives us that $\interp{A}$ is a
      saturated set, so $a \in \SN$.  We already observed that
      $\subst{A} \in \SN$.  By the second IH, $\subst{b} = \subst[
      \rho[x \mapsto x\rbracket ]{b} \in \interp{B}$.  But the
      classification lemma and lemma~\ref{lem:interptype} imply that
      $\interp{B}$ is a saturated set, so $\subst{b} \in \SN$.  Thus,
      by lemma~\ref{lem:keyredexSN}, $(\lambda x:\subst{A}
      . \subst{b})\;a \in \SN$ as desired.

    \item Suppose instead that $s$ is $ \Box $.  We must show
      $\lambda x : \subst{A} . \subst{b} \in
      \sarr{\interp{A}}{\displaystyle\bigcap_{S \in V(A)}
        \interp[\ext{\sigma}{x}{S}]{B}}$.  Let an expression $a \in
      \interp{A}$ and a saturated set $S \in V(A)$ be given.  It is
      enough to show $(\lambda x:\subst{A} . \subst{b})\;a \in
      \interp[\ext{\sigma}{x}{S}]{B}$.

      Because $\ext{\sigma}{x}{S} \models \ext{\rho}{x}{a} :
       \Gamma , \ottmv{x} \!:\! \ottnt{A} $, the IH for $\D_2$ gives us that
      $\subst[\ext{\rho}{x}{a}]{b} \in
      \interp[\ext{\sigma}{x}{S}]{B}$.  As in the previous subcase, we
      can observe that:
      $$
      (\lambda x:\subst{A} . \subst{b})\;a  \leadsto 
      [a/x] \subst{b} = \subst[\ext{\rho}{x}{a}]{b}
      $$
      This step contracts a key redex, and by reasoning as before we
      find $(\lambda x:\subst{A} . \subst{b})\;a \in
      \interp[\ext{\sigma}{x}{S}]{B}$ as desired.
    \end{itemize}

  \item[{\bf Case:}] $\D = \lowered {
     \infer[App]
       {\infer{}{\D_1 \\\\ \Gamma  \vdash  \ottnt{a}  \ottsym{:}  \ottsym{(}  \ottmv{x}  \ottsym{:}  \ottnt{A}  \ottsym{)}  \to  \ottnt{B}}\\
        \infer{}{\D_2 \\\\ \Gamma  \vdash  \ottnt{b}  \ottsym{:}  \ottnt{A}}}
       {\Gamma  \vdash  \ottnt{a} \, \ottnt{b}  \ottsym{:}  \ottsym{[}  \ottnt{b}  \ottsym{/}  \ottmv{x}  \ottsym{]}  \ottnt{B}}
    }$

    We must show $\subst{a}\;\subst{b} \in \interp{[b/x]B}$, and the
    IH for $\D_2$ is that $\subst{b} \in \interp{A}$.  We consider two
    cases: by the classification lemma and inversion, $A$ is either a
    kind or a $\Gamma$-type.
 
    \begin{itemize}
    \item Suppose first that $A$ is a kind.  Then the IH for $\D_1$
      gives us that
      $$
      \subst{a} \in 
        \sarr{\interp{A}}{\displaystyle\bigcap_{S \in V(A)}
                             \interp[\ext{\sigma}{x}{S}]{B}}
      $$
      In particular, expanding the definition of $\sarr{\cdot}{\cdot}$
      and applying lemma~\ref{lem:interptype}, we have:
      $$
      \subst{a}\;\subst{b} \in \interp[\ext{\sigma}{x}{\interp{b}}]{B}
      $$
      By lemma~\ref{lem:interpsubst}, this is just what we wanted to
      show.
 
    \item The case where $A$ is a $\Gamma$-type is similar to but
      slightly simpler than the last case.
    \end{itemize}
   
  \item[{\bf Case:}] $\D = \lowered {
    \infer[Conv]
      {\infer{}{\D_1 \\\\ \Gamma  \vdash  \ottnt{a}  \ottsym{:}  \ottnt{A}} \\
       \infer{}{\D_2 \\\\ \Gamma  \vdash  \ottnt{B}  \ottsym{:}  \ottnt{s}} \\
       \ottnt{A}  =_{\beta}  \ottnt{B}}
     {\Gamma  \vdash  \ottnt{a}  \ottsym{:}  \ottnt{B}}
    }$

    The IH for $\D_1$ gives us that $\subst{a} \in \interp{A}$.  By
    the classification lemma and $\D_2$, both $A$ and $B$ are kinds or
    $\Gamma$-constructors.  Thus, by lemma~\ref{lem:interpbeta},
    $\subst{a} \in \interp{B}$ as desired.  \qedhere
  \end{itemize}
\end{proof}

\end{document}

%% file: wpe_inc.tex
\newcommand{\ottdrule}[4][]{{\displaystyle\frac{\begin{array}{l}#2\end{array}}{#3}\quad\ottdrulename{#4}}}
\newcommand{\ottusedrule}[1]{\[#1\]}
\newcommand{\ottpremise}[1]{ #1 \\}
\newenvironment{ottdefnblock}[3][]{ \framebox{\mbox{#2}} \quad #3 \\[0pt]}{}
\newenvironment{ottfundefnblock}[3][]{ \framebox{\mbox{#2}} \quad #3 \\[0pt]\begin{displaymath}\begin{array}{l}}{\end{array}\end{displaymath}}
\newcommand{\ottfunclause}[2]{ #1 \equiv #2 \\}
\newcommand{\ottnt}[1]{\mathit{#1}}
\newcommand{\ottmv}[1]{\mathit{#1}}
\newcommand{\ottkw}[1]{\mathbf{#1}}
\newcommand{\ottsym}[1]{#1}
\newcommand{\ottcom}[1]{\text{#1}}
\newcommand{\ottdrulename}[1]{\textsc{#1}}
\newcommand{\ottcomplu}[5]{\overline{#1}^{\,#2\in #3 #4 #5}}
\newcommand{\ottcompu}[3]{\overline{#1}^{\,#2<#3}}
\newcommand{\ottcomp}[2]{\overline{#1}^{\,#2}}
\newcommand{\ottgrammartabular}[1]{\begin{supertabular}{llcllllll}#1\end{supertabular}}
\newcommand{\ottmetavartabular}[1]{\begin{supertabular}{ll}#1\end{supertabular}}
\newcommand{\ottrulehead}[3]{$#1$ & & $#2$ & & & \multicolumn{2}{l}{#3}}
\newcommand{\ottprodline}[6]{& & $#1$ & $#2$ & $#3 #4$ & $#5$ & $#6$}
\newcommand{\ottfirstprodline}[6]{\ottprodline{#1}{#2}{#3}{#4}{#5}{#6}}
\newcommand{\ottlongprodline}[2]{& & $#1$ & \multicolumn{4}{l}{$#2$}}
\newcommand{\ottfirstlongprodline}[2]{\ottlongprodline{#1}{#2}}
\newcommand{\ottbindspecprodline}[6]{\ottprodline{#1}{#2}{#3}{#4}{#5}{#6}}
\newcommand{\ottprodnewline}{\\}
\newcommand{\ottinterrule}{\\[5.0mm]}
\newcommand{\ottafterlastrule}{\\}
\newcommand{\ottmetavars}{
\ottmetavartabular{
 $ \ottmv{var} ,\, \ottmv{x} ,\, \ottmv{y} ,\, \ottmv{z} ,\, \ottmv{w} $ &  \\
 $ \ottmv{n} $ &  \\
}}

\newcommand{\otts}{
\ottrulehead{\ottnt{s}}{::=}{}}

\newcommand{\ottexp}{
\ottrulehead{\ottnt{exp}  ,\ \ottnt{A}  ,\ \ottnt{B}  ,\ \ottnt{F}  ,\ \ottnt{a}  ,\ \ottnt{b}}{::=}{}\ottprodnewline
\ottfirstprodline{|}{\ottnt{s}}{}{}{}{}\ottprodnewline
\ottprodline{|}{\ottmv{x}}{}{}{}{}\ottprodnewline
\ottprodline{|}{\ottsym{(}  \ottmv{x}  \ottsym{:}  \ottnt{A}  \ottsym{)}  \to  \ottnt{B}}{}{\textsf{bind}\; \ottmv{x}\; \textsf{in}\; \ottnt{B}}{}{}\ottprodnewline
\ottprodline{|}{\lambda  \ottmv{x}  \ottsym{:}  \ottnt{A}  \ottsym{.}  \ottnt{b}}{}{\textsf{bind}\; \ottmv{x}\; \textsf{in}\; \ottnt{b}}{}{}\ottprodnewline
\ottprodline{|}{\ottnt{a} \, \ottnt{b}}{}{}{}{}}

\newcommand{\ottctx}{
\ottrulehead{\ottnt{ctx}  ,\ \Gamma}{::=}{}\ottprodnewline
\ottfirstprodline{|}{ \cdot }{}{}{}{}\ottprodnewline
\ottprodline{|}{ \Gamma , \ottmv{x} \!:\! \ottnt{A} }{}{}{}{}\ottprodnewline
\ottprodline{|}{ \Gamma ,w^{ \ottmv{x} }\!:\! \ottnt{A} }{}{}{}{}\ottprodnewline
\ottprodline{|}{ \fomCtxI{ \Gamma } }{}{}{}{}}

\newcommand{\otttexp}{
\ottrulehead{exp  ,\ A  ,\ B  ,\ a  ,\ b}{::=}{}\ottprodnewline
\ottfirstprodline{|}{\ottnt{s}}{}{}{}{}\ottprodnewline
\ottprodline{|}{\ottmv{x}}{}{}{}{}\ottprodnewline
\ottprodline{|}{\ottsym{(}  \ottmv{x}  \ottsym{:}  A  \ottsym{)}  \to  B}{}{\textsf{bind}\; \ottmv{x}\; \textsf{in}\; B}{}{}\ottprodnewline
\ottprodline{|}{ \lambda _{( \ottmv{x} : A ) \to B }. b }{}{\textsf{bind}\; \ottmv{x}\; \textsf{in}\; b}{}{}\ottprodnewline
\ottbindspecprodline{}{}{}{\textsf{bind}\; \ottmv{x}\; \textsf{in}\; B}{}{}\ottprodnewline
\ottprodline{|}{ \mathsf{app} _{( \ottmv{x} : A ) \to B } ( a , b ) }{}{}{}{}}

\newcommand{\otttctx}{
\ottrulehead{ctx  ,\ \Gamma}{::=}{}\ottprodnewline
\ottfirstprodline{|}{ \cdot }{}{}{}{}\ottprodnewline
\ottprodline{|}{ \Gamma , \ottmv{x} \!:\! A }{}{}{}{}}

\newcommand{\ottgrammar}{\ottgrammartabular{
\otts\ottinterrule
\ottexp\ottinterrule
\ottctx\ottinterrule
\otttexp\ottinterrule
\otttctx\ottafterlastrule
}}

\newcommand{\ottdruleSBeta}[1]{\ottdrule[#1]{%
}{
\ottsym{(}  \lambda  \ottmv{x}  \ottsym{:}  \ottnt{A}  \ottsym{.}  \ottnt{b}  \ottsym{)} \, \ottnt{a}  \leadsto  \ottsym{[}  \ottnt{a}  \ottsym{/}  \ottmv{x}  \ottsym{]}  \ottnt{b}}{%
{\ottdrulename{SBeta}}{}%
}}

\newcommand{\ottdruleSPiOne}[1]{\ottdrule[#1]{%
\ottpremise{\ottnt{A}  \leadsto  \ottnt{A'}}%
}{
\ottsym{(}  \ottmv{x}  \ottsym{:}  \ottnt{A}  \ottsym{)}  \to  \ottnt{B}  \leadsto  \ottsym{(}  \ottmv{x}  \ottsym{:}  \ottnt{A'}  \ottsym{)}  \to  \ottnt{B}}{%
{\ottdrulename{SPi1}}{}%
}}

\newcommand{\ottdruleSPiTwo}[1]{\ottdrule[#1]{%
\ottpremise{\ottnt{B}  \leadsto  \ottnt{B'}}%
}{
\ottsym{(}  \ottmv{x}  \ottsym{:}  \ottnt{A}  \ottsym{)}  \to  \ottnt{B}  \leadsto  \ottsym{(}  \ottmv{x}  \ottsym{:}  \ottnt{A}  \ottsym{)}  \to  \ottnt{B'}}{%
{\ottdrulename{SPi2}}{}%
}}

\newcommand{\ottdruleSLamOne}[1]{\ottdrule[#1]{%
\ottpremise{\ottnt{A}  \leadsto  \ottnt{A'}}%
}{
\lambda  \ottmv{x}  \ottsym{:}  \ottnt{A}  \ottsym{.}  \ottnt{b}  \leadsto  \lambda  \ottmv{x}  \ottsym{:}  \ottnt{A'}  \ottsym{.}  \ottnt{b}}{%
{\ottdrulename{SLam1}}{}%
}}

\newcommand{\ottdruleSLamTwo}[1]{\ottdrule[#1]{%
\ottpremise{\ottnt{b}  \leadsto  \ottnt{b'}}%
}{
\lambda  \ottmv{x}  \ottsym{:}  \ottnt{A}  \ottsym{.}  \ottnt{b}  \leadsto  \lambda  \ottmv{x}  \ottsym{:}  \ottnt{A}  \ottsym{.}  \ottnt{b'}}{%
{\ottdrulename{SLam2}}{}%
}}

\newcommand{\ottdruleSAppOne}[1]{\ottdrule[#1]{%
\ottpremise{\ottnt{a}  \leadsto  \ottnt{a'}}%
}{
\ottnt{a} \, \ottnt{b}  \leadsto  \ottnt{a'} \, \ottnt{b}}{%
{\ottdrulename{SApp1}}{}%
}}

\newcommand{\ottdruleSAppTwo}[1]{\ottdrule[#1]{%
\ottpremise{\ottnt{b}  \leadsto  \ottnt{b'}}%
}{
\ottnt{a} \, \ottnt{b}  \leadsto  \ottnt{a} \, \ottnt{b'}}{%
{\ottdrulename{SApp2}}{}%
}}

\newcommand{\ottdefnstep}[1]{\begin{ottdefnblock}[#1]{$\ottnt{a}  \leadsto  \ottnt{b}$}{}
\ottusedrule{\ottdruleSBeta{}}
\ottusedrule{\ottdruleSPiOne{}}
\ottusedrule{\ottdruleSPiTwo{}}
\ottusedrule{\ottdruleSLamOne{}}
\ottusedrule{\ottdruleSLamTwo{}}
\ottusedrule{\ottdruleSAppOne{}}
\ottusedrule{\ottdruleSAppTwo{}}
\end{ottdefnblock}}

\newcommand{\ottdefnsJstep}{
\ottdefnstep{}}

\newcommand{\ottdruleMSRefl}[1]{\ottdrule[#1]{%
}{
\ottnt{a}  \leadsto^{*}  \ottnt{a}}{%
{\ottdrulename{MSRefl}}{}%
}}

\newcommand{\ottdruleMSStep}[1]{\ottdrule[#1]{%
\ottpremise{\ottnt{a_{{\mathrm{1}}}}  \leadsto  \ottnt{a_{{\mathrm{2}}}}}%
\ottpremise{\ottnt{a_{{\mathrm{2}}}}  \leadsto^{*}  \ottnt{a_{{\mathrm{3}}}}}%
}{
\ottnt{a_{{\mathrm{1}}}}  \leadsto^{*}  \ottnt{a_{{\mathrm{3}}}}}{%
{\ottdrulename{MSStep}}{}%
}}

\newcommand{\ottdefnmstep}[1]{\begin{ottdefnblock}[#1]{$\ottnt{a}  \leadsto^{*}  \ottnt{b}$}{}
\ottusedrule{\ottdruleMSRefl{}}
\ottusedrule{\ottdruleMSStep{}}
\end{ottdefnblock}}

\newcommand{\ottdefnsJMstep}{
\ottdefnmstep{}}

\newcommand{\ottdruleCVRefl}[1]{\ottdrule[#1]{%
}{
\ottnt{a}  =_{\beta}  \ottnt{a}}{%
{\ottdrulename{CVRefl}}{}%
}}

\newcommand{\ottdruleCVExp}[1]{\ottdrule[#1]{%
\ottpremise{\ottnt{a_{{\mathrm{1}}}}  =_{\beta}  \ottnt{a_{{\mathrm{2}}}}}%
\ottpremise{\ottnt{a_{{\mathrm{3}}}}  \leadsto  \ottnt{a_{{\mathrm{2}}}}}%
}{
\ottnt{a_{{\mathrm{1}}}}  =_{\beta}  \ottnt{a_{{\mathrm{3}}}}}{%
{\ottdrulename{CVExp}}{}%
}}

\newcommand{\ottdruleCVRed}[1]{\ottdrule[#1]{%
\ottpremise{\ottnt{a_{{\mathrm{1}}}}  =_{\beta}  \ottnt{a_{{\mathrm{2}}}}}%
\ottpremise{\ottnt{a_{{\mathrm{2}}}}  \leadsto  \ottnt{a_{{\mathrm{3}}}}}%
}{
\ottnt{a_{{\mathrm{1}}}}  =_{\beta}  \ottnt{a_{{\mathrm{3}}}}}{%
{\ottdrulename{CVRed}}{}%
}}

\newcommand{\ottdefnconv}[1]{\begin{ottdefnblock}[#1]{$\ottnt{a}  =_{\beta}  \ottnt{b}$}{}
\ottusedrule{\ottdruleCVRefl{}}
\ottusedrule{\ottdruleCVExp{}}
\ottusedrule{\ottdruleCVRed{}}
\end{ottdefnblock}}

\newcommand{\ottdefnsJConv}{
\ottdefnconv{}}

\newcommand{\ottdruleTSort}[1]{\ottdrule[#1]{%
\ottpremise{\vdash  \Gamma}%
\ottpremise{ ( \ottnt{s_{{\mathrm{1}}}} , \ottnt{s_{{\mathrm{2}}}} ) \in \mathcal{A} }%
}{
\Gamma  \vdash  \ottnt{s_{{\mathrm{1}}}}  \ottsym{:}  \ottnt{s_{{\mathrm{2}}}}}{%
{\ottdrulename{TSort}}{}%
}}

\newcommand{\ottdruleTVar}[1]{\ottdrule[#1]{%
\ottpremise{\vdash  \Gamma}%
\ottpremise{ ( \ottmv{x} : \ottnt{A} ) \in  \Gamma }%
}{
\Gamma  \vdash  \ottmv{x}  \ottsym{:}  \ottnt{A}}{%
{\ottdrulename{TVar}}{}%
}}

\newcommand{\ottdruleTPi}[1]{\ottdrule[#1]{%
\ottpremise{\Gamma  \vdash  \ottnt{A}  \ottsym{:}  \ottnt{s_{{\mathrm{1}}}}}%
\ottpremise{ \Gamma , \ottmv{x} \!:\! \ottnt{A}   \vdash  \ottnt{B}  \ottsym{:}  \ottnt{s_{{\mathrm{2}}}}}%
\ottpremise{ ( \ottnt{s_{{\mathrm{1}}}} , \ottnt{s_{{\mathrm{2}}}} , \ottnt{s_{{\mathrm{3}}}} ) \in \mathcal{R} }%
}{
\Gamma  \vdash  \ottsym{(}  \ottmv{x}  \ottsym{:}  \ottnt{A}  \ottsym{)}  \to  \ottnt{B}  \ottsym{:}  \ottnt{s_{{\mathrm{3}}}}}{%
{\ottdrulename{TPi}}{}%
}}

\newcommand{\ottdruleTLam}[1]{\ottdrule[#1]{%
\ottpremise{\Gamma  \vdash  \ottsym{(}  \ottmv{x}  \ottsym{:}  \ottnt{A}  \ottsym{)}  \to  \ottnt{B}  \ottsym{:}  \ottnt{s}}%
\ottpremise{ \Gamma , \ottmv{x} \!:\! \ottnt{A}   \vdash  \ottnt{b}  \ottsym{:}  \ottnt{B}}%
}{
\Gamma  \vdash  \lambda  \ottmv{x}  \ottsym{:}  \ottnt{A}  \ottsym{.}  \ottnt{b}  \ottsym{:}  \ottsym{(}  \ottmv{x}  \ottsym{:}  \ottnt{A}  \ottsym{)}  \to  \ottnt{B}}{%
{\ottdrulename{TLam}}{}%
}}

\newcommand{\ottdruleTApp}[1]{\ottdrule[#1]{%
\ottpremise{\Gamma  \vdash  \ottnt{a}  \ottsym{:}  \ottsym{(}  \ottmv{x}  \ottsym{:}  \ottnt{A}  \ottsym{)}  \to  \ottnt{B}}%
\ottpremise{\Gamma  \vdash  \ottnt{b}  \ottsym{:}  \ottnt{A}}%
}{
\Gamma  \vdash  \ottnt{a} \, \ottnt{b}  \ottsym{:}  \ottsym{[}  \ottnt{b}  \ottsym{/}  \ottmv{x}  \ottsym{]}  \ottnt{B}}{%
{\ottdrulename{TApp}}{}%
}}

\newcommand{\ottdruleTConv}[1]{\ottdrule[#1]{%
\ottpremise{ \Gamma  \vdash  \ottnt{a}  \ottsym{:}  \ottnt{A}  \qquad  \Gamma  \vdash  \ottnt{B}  \ottsym{:}  \ottnt{s} }%
\ottpremise{\ottnt{A}  =_{\beta}  \ottnt{B}}%
}{
\Gamma  \vdash  \ottnt{a}  \ottsym{:}  \ottnt{B}}{%
{\ottdrulename{TConv}}{}%
}}

\newcommand{\ottdefntyp}[1]{\begin{ottdefnblock}[#1]{$\Gamma  \vdash  \ottnt{a}  \ottsym{:}  \ottnt{A}$}{}
\ottusedrule{\ottdruleTSort{}}
\ottusedrule{\ottdruleTVar{}}
\ottusedrule{\ottdruleTPi{}}
\ottusedrule{\ottdruleTLam{}}
\ottusedrule{\ottdruleTApp{}}
\ottusedrule{\ottdruleTConv{}}
\end{ottdefnblock}}

\newcommand{\ottdruleCNil}[1]{\ottdrule[#1]{%
}{
\vdash   \cdot }{%
{\ottdrulename{CNil}}{}%
}}

\newcommand{\ottdruleCCons}[1]{\ottdrule[#1]{%
\ottpremise{ \ottmv{x}  \notin \mathsf{dom}( \Gamma ) }%
\ottpremise{\Gamma  \vdash  \ottnt{A}  \ottsym{:}  \ottnt{s}}%
}{
\vdash   \Gamma , \ottmv{x} \!:\! \ottnt{A} }{%
{\ottdrulename{CCons}}{}%
}}

\newcommand{\ottdefnctyp}[1]{\begin{ottdefnblock}[#1]{$\vdash  \Gamma$}{}
\ottusedrule{\ottdruleCNil{}}
\ottusedrule{\ottdruleCCons{}}
\end{ottdefnblock}}

\newcommand{\ottdefnsJTyp}{
\ottdefntyp{}\ottdefnctyp{}}

\newcommand{\ottdruleTSBeta}[1]{\ottdrule[#1]{%
}{
 \mathsf{app} _{( \ottmv{x} : A ) \to B } ( \ottsym{(}   \lambda _{( \ottmv{x} : A ) \to B }. b   \ottsym{)} , a )   \leadsto_{t}  \ottsym{[}  a  \ottsym{/}  \ottmv{x}  \ottsym{]}  b}{%
{\ottdrulename{TSBeta}}{}%
}}

\newcommand{\ottdruleTSPiOne}[1]{\ottdrule[#1]{%
\ottpremise{A  \leadsto_{t}  A'}%
}{
\ottsym{(}  \ottmv{x}  \ottsym{:}  A  \ottsym{)}  \to  B  \leadsto_{t}  \ottsym{(}  \ottmv{x}  \ottsym{:}  A'  \ottsym{)}  \to  B}{%
{\ottdrulename{TSPi1}}{}%
}}

\newcommand{\ottdruleTSPiTwo}[1]{\ottdrule[#1]{%
\ottpremise{B  \leadsto_{t}  B'}%
}{
\ottsym{(}  \ottmv{x}  \ottsym{:}  A  \ottsym{)}  \to  B  \leadsto_{t}  \ottsym{(}  \ottmv{x}  \ottsym{:}  A  \ottsym{)}  \to  B'}{%
{\ottdrulename{TSPi2}}{}%
}}

\newcommand{\ottdruleTSLamOne}[1]{\ottdrule[#1]{%
\ottpremise{A  \leadsto_{t}  A'}%
}{
 \lambda _{( \ottmv{x} : A ) \to B }. b   \leadsto_{t}   \lambda _{( \ottmv{x} : A' ) \to B }. b }{%
{\ottdrulename{TSLam1}}{}%
}}

\newcommand{\ottdruleTSLamTwo}[1]{\ottdrule[#1]{%
\ottpremise{B  \leadsto_{t}  B'}%
}{
 \lambda _{( \ottmv{x} : A ) \to B }. b   \leadsto_{t}   \lambda _{( \ottmv{x} : A ) \to B }. b }{%
{\ottdrulename{TSLam2}}{}%
}}

\newcommand{\ottdruleTSLamThree}[1]{\ottdrule[#1]{%
\ottpremise{b  \leadsto_{t}  b'}%
}{
 \lambda _{( \ottmv{x} : A ) \to B }. b   \leadsto_{t}   \lambda _{( \ottmv{x} : A ) \to B }. b' }{%
{\ottdrulename{TSLam3}}{}%
}}

\newcommand{\ottdruleTSAppOne}[1]{\ottdrule[#1]{%
\ottpremise{A  \leadsto_{t}  A'}%
}{
 \mathsf{app} _{( \ottmv{x} : A ) \to B } ( a , b )   \leadsto_{t}   \mathsf{app} _{( \ottmv{x} : A' ) \to B } ( a , b ) }{%
{\ottdrulename{TSApp1}}{}%
}}

\newcommand{\ottdruleTSAppTwo}[1]{\ottdrule[#1]{%
\ottpremise{B  \leadsto_{t}  B'}%
}{
 \mathsf{app} _{( \ottmv{x} : A ) \to B } ( a , b )   \leadsto_{t}   \mathsf{app} _{( \ottmv{x} : A ) \to B' } ( a , b ) }{%
{\ottdrulename{TSApp2}}{}%
}}

\newcommand{\ottdruleTSAppThree}[1]{\ottdrule[#1]{%
\ottpremise{a  \leadsto_{t}  a'}%
}{
 \mathsf{app} _{( \ottmv{x} : A ) \to B } ( a , b )   \leadsto_{t}   \mathsf{app} _{( \ottmv{x} : A ) \to B } ( a' , b ) }{%
{\ottdrulename{TSApp3}}{}%
}}

\newcommand{\ottdruleTSAppFour}[1]{\ottdrule[#1]{%
\ottpremise{b  \leadsto_{t}  b'}%
}{
 \mathsf{app} _{( \ottmv{x} : A ) \to B } ( a , b )   \leadsto_{t}   \mathsf{app} _{( \ottmv{x} : A ) \to B } ( a , b' ) }{%
{\ottdrulename{TSApp4}}{}%
}}

\newcommand{\ottdefntstep}[1]{\begin{ottdefnblock}[#1]{$a  \leadsto_{t}  b$}{}
\ottusedrule{\ottdruleTSBeta{}}
\ottusedrule{\ottdruleTSPiOne{}}
\ottusedrule{\ottdruleTSPiTwo{}}
\ottusedrule{\ottdruleTSLamOne{}}
\ottusedrule{\ottdruleTSLamTwo{}}
\ottusedrule{\ottdruleTSLamThree{}}
\ottusedrule{\ottdruleTSAppOne{}}
\ottusedrule{\ottdruleTSAppTwo{}}
\ottusedrule{\ottdruleTSAppThree{}}
\ottusedrule{\ottdruleTSAppFour{}}
\end{ottdefnblock}}

\newcommand{\ottdefnsJTstep}{
\ottdefntstep{}}

\newcommand{\ottdruleMTSRefl}[1]{\ottdrule[#1]{%
}{
a  \leadsto^{*}_{t}  a}{%
{\ottdrulename{MTSRefl}}{}%
}}

\newcommand{\ottdruleMTSStep}[1]{\ottdrule[#1]{%
\ottpremise{a_{{\mathrm{1}}}  \leadsto_{t}  a_{{\mathrm{2}}}}%
\ottpremise{a_{{\mathrm{2}}}  \leadsto^{*}_{t}  a_{{\mathrm{3}}}}%
}{
a_{{\mathrm{1}}}  \leadsto^{*}_{t}  a_{{\mathrm{3}}}}{%
{\ottdrulename{MTSStep}}{}%
}}

\newcommand{\ottdefnmtstep}[1]{\begin{ottdefnblock}[#1]{$a  \leadsto^{*}_{t}  b$}{}
\ottusedrule{\ottdruleMTSRefl{}}
\ottusedrule{\ottdruleMTSStep{}}
\end{ottdefnblock}}

\newcommand{\ottdefnsJMTstep}{
\ottdefnmtstep{}}

\newcommand{\ottdruleTTSort}[1]{\ottdrule[#1]{%
\ottpremise{\vdash_{t}  \Gamma}%
\ottpremise{ ( \ottnt{s_{{\mathrm{1}}}} , \ottnt{s_{{\mathrm{2}}}} ) \in \mathcal{A} }%
}{
\Gamma  \vdash_{t}  \ottnt{s_{{\mathrm{1}}}}  \ottsym{:}  \ottnt{s_{{\mathrm{2}}}}}{%
{\ottdrulename{TTSort}}{}%
}}

\newcommand{\ottdruleTTVar}[1]{\ottdrule[#1]{%
\ottpremise{\vdash_{t}  \Gamma}%
\ottpremise{ ( \ottmv{x} : A ) \in  \Gamma }%
}{
\Gamma  \vdash_{t}  \ottmv{x}  \ottsym{:}  A}{%
{\ottdrulename{TTVar}}{}%
}}

\newcommand{\ottdruleTTPi}[1]{\ottdrule[#1]{%
\ottpremise{\Gamma  \vdash_{t}  A  \ottsym{:}  \ottnt{s_{{\mathrm{1}}}}}%
\ottpremise{ \Gamma , \ottmv{x} \!:\! A   \vdash_{t}  B  \ottsym{:}  \ottnt{s_{{\mathrm{2}}}}}%
\ottpremise{ ( \ottnt{s_{{\mathrm{1}}}} , \ottnt{s_{{\mathrm{2}}}} , \ottnt{s_{{\mathrm{3}}}} ) \in \mathcal{R} }%
}{
\Gamma  \vdash_{t}  \ottsym{(}  \ottmv{x}  \ottsym{:}  A  \ottsym{)}  \to  B  \ottsym{:}  \ottnt{s_{{\mathrm{3}}}}}{%
{\ottdrulename{TTPi}}{}%
}}

\newcommand{\ottdruleTTLam}[1]{\ottdrule[#1]{%
\ottpremise{\Gamma  \vdash_{t}  \ottsym{(}  \ottmv{x}  \ottsym{:}  A  \ottsym{)}  \to  B  \ottsym{:}  \ottnt{s}}%
\ottpremise{ \Gamma , \ottmv{x} \!:\! A   \vdash_{t}  b  \ottsym{:}  B}%
}{
\Gamma  \vdash_{t}   \lambda _{( \ottmv{x} : A ) \to B }. b   \ottsym{:}  \ottsym{(}  \ottmv{x}  \ottsym{:}  A  \ottsym{)}  \to  B}{%
{\ottdrulename{TTLam}}{}%
}}

\newcommand{\ottdruleTTApp}[1]{\ottdrule[#1]{%
\ottpremise{\Gamma  \vdash_{t}  a  \ottsym{:}  \ottsym{(}  \ottmv{x}  \ottsym{:}  A  \ottsym{)}  \to  B}%
\ottpremise{\Gamma  \vdash_{t}  b  \ottsym{:}  A}%
}{
\Gamma  \vdash_{t}   \mathsf{app} _{( \ottmv{x} : A ) \to B } ( a , b )   \ottsym{:}  \ottsym{[}  b  \ottsym{/}  \ottmv{x}  \ottsym{]}  B}{%
{\ottdrulename{TTApp}}{}%
}}

\newcommand{\ottdruleTTConv}[1]{\ottdrule[#1]{%
\ottpremise{ \Gamma  \vdash_{t}  a  \ottsym{:}  A  \qquad  \Gamma  \vdash_{t}  B  \ottsym{:}  \ottnt{s} }%
\ottpremise{ A  \leadsto^{*}_{t}  B \hspace{1.5em}\text{or}\hspace{1.5em} B  \leadsto^{*}_{t}  A }%
}{
\Gamma  \vdash_{t}  a  \ottsym{:}  B}{%
{\ottdrulename{TTConv}}{}%
}}

\newcommand{\ottdefnttyp}[1]{\begin{ottdefnblock}[#1]{$\Gamma  \vdash_{t}  a  \ottsym{:}  A$}{}
\ottusedrule{\ottdruleTTSort{}}
\ottusedrule{\ottdruleTTVar{}}
\ottusedrule{\ottdruleTTPi{}}
\ottusedrule{\ottdruleTTLam{}}
\ottusedrule{\ottdruleTTApp{}}
\ottusedrule{\ottdruleTTConv{}}
\end{ottdefnblock}}

\newcommand{\ottdruleTCNil}[1]{\ottdrule[#1]{%
}{
\vdash_{t}   \cdot }{%
{\ottdrulename{TCNil}}{}%
}}

\newcommand{\ottdruleTCCons}[1]{\ottdrule[#1]{%
\ottpremise{ \ottmv{x}  \notin \mathsf{dom}( \Gamma ) }%
\ottpremise{\Gamma  \vdash_{t}  A  \ottsym{:}  \ottnt{s}}%
}{
\vdash_{t}   \Gamma , \ottmv{x} \!:\! A }{%
{\ottdrulename{TCCons}}{}%
}}

\newcommand{\ottdefntctyp}[1]{\begin{ottdefnblock}[#1]{$\vdash_{t}  \Gamma$}{}
\ottusedrule{\ottdruleTCNil{}}
\ottusedrule{\ottdruleTCCons{}}
\end{ottdefnblock}}

\newcommand{\ottdefnsJTTyp}{
\ottdefnttyp{}\ottdefntctyp{}}

\newcommand{\ottdefnss}{
\ottdefnsJstep
\ottdefnsJMstep
\ottdefnsJConv
\ottdefnsJTyp
\ottdefnsJTstep
\ottdefnsJMTstep
\ottdefnsJTTyp
}

\newcommand{\ottall}{\ottmetavars\\[0pt]
\ottgrammar\\[5.0mm]
\ottdefnss}